\documentclass[11pt,reqno,oneside]{amsart}

\usepackage{graphicx}
\usepackage{amssymb}
\usepackage{epstopdf}

\usepackage[a4paper, total={6in, 9.1in}]{geometry}

\usepackage{amsmath,amsfonts,amsthm,mathrsfs,amssymb,cite}
\usepackage[usenames]{color}

\usepackage{subfigure} 
\usepackage{framed}

\usepackage{enumerate}
\usepackage{hyperref}
\usepackage{xcolor}
\hypersetup{
  colorlinks,
  linkcolor={blue!80!black},
  urlcolor={blue!80!black},
  citecolor={blue!80!black}
}
\usepackage[capitalize,noabbrev]{cleveref}

\usepackage{booktabs} 
\usepackage{array} 
\usepackage{paralist} 
\usepackage{verbatim} 
\usepackage{tabularx}
\usepackage{amsmath,amsfonts,amsthm,mathrsfs,amssymb,cite}
\usepackage[usenames]{color}
\usepackage{bm}

\newtheorem{thm}{Theorem}[section]

\newtheorem{lem}{Lemma}[section]
\newtheorem{prop}{Proposition}[section]
\theoremstyle{definition}
\newtheorem{defn}{Definition}[section]
\theoremstyle{remark}

\newtheorem{rem}{Remark}[section]
\numberwithin{equation}{section}

\newcommand{\abs}[1]{\left\vert#1\right\vert}

\newcommand{\bu}{\mathbf{u}}
\newcommand{\bv}{\mathbf{v}}

\newcommand{\bw}{\mathbf{w}}
\newcommand{\bff}{\mathbf{f}}
\newcommand{\bvarphi}{\bm{\varphi}}
\newcommand{\bpsi}{\bm{\psi}}
\newcommand{\bnu}{\bm{\nu}}
\newcommand{\bGa}{\mathbf{\Gamma}}
\newcommand{\bx}{\mathbf{x}}
\newcommand{\by}{\mathbf{y}}
\newcommand{\rmi}{\mathrm{i}}
\newcommand{\bS}{\mathbf{S}}
\newcommand{\bK}{\mathbf{K}}
\newcommand{\bI}{\mathbf{I}}
\newcommand{\bbS}{\mathbb{S}}
\newcommand{\bA}{\mathbf{A}}

\newcommand{\bb}{\mathbf{b}}

\newcommand{\bF}{\mathbf{F}}
\newcommand{\bt}{\mathbf{t}}

\newcommand{\Lcal}{\mathcal{L}}

\newcommand{\Ocal}{\mathcal{O}}
\newcommand{\Rcal}{\mathcal{R}}
\newcommand{\Tcal}{\mathcal{T}}

\def\ms{\medskip}

\allowdisplaybreaks

\title[Elastodynamical resonances and cloaking]{Elastodynamical resonances and cloaking of negative material structures beyond quasistatic approximation}

\author{Hongjie Li}
\address{Department of Mathematics, The Chinese University of Hong Kong, Shatin, Hong Kong SAR, China.}
\email{hjli@math.cuhk.edu.hk}

\author{Hongyu Liu}
\address{Department of Mathematics, City University of Hong Kong, Kowloon, Hong Kong SAR, China.}
\email{hongyliu@cityu.edu.hk; hongyu.liuip@gmail.com}

\author{Jun Zou}
\address{Department of Mathematics, The Chinese University of Hong Kong, Shatin, Hong Kong SAR, China.}
\email{zou@math.cuhk.edu.hk}

\begin{document}
\maketitle

\begin{abstract}

Given the flexibility of choosing negative elastic parameters, we construct material structures that can induce two resonance phenomena, {referred to as the elastodynamical resonances}. They mimic the emerging plasmon/polariton resonance and anomalous localized resonance in optics for subwavelength particles. However, we study the peculiar resonance phenomena for linear elasticity beyond 
the subwavelength regime. It is shown that the resonance behaviours possess distinct characters, with some similar to the subwavelength resonances, but some sharply different due to the frequency effect. It is particularly noted that we construct a core-shell material structure that can induce anomalous localized resonance as well as cloaking phenomena beyond the quasi-static limit. The study is boiled down to analyzing the so-called elastic Neumann-Poinc\'are (N-P) operator in the frequency regime. We provide an in-depth analysis of the spectral properties of the N-P operator beyond the quasi-static approximation, 
{and these results are of independent interest to the spectral theory of layer potential operators.} 

\medskip

\noindent{\bf Keywords:}~~negative materials,  core-shell structure, anomalous localized resonance, beyond quasistatic limit, Neumann-Poinc\'are operator, spectral, invisibility cloaking

\ms

\end{abstract}

\section{Introduction}\label{sec:introduction} 

\subsection{Mathematical formulation and main findings }
{We initially focus on the mathematics, not on the physics, and}
present the Lam\'e system which governs the propagation of linear elastic deformation.

For $\bx\in\mathbb{R}^N$, $N=2,3$, 
we write $\mathbf{C}(\bx):=(\mathrm{C}_{ijkl}(\bx))_{i,j,k,l=1}^N$ as a four-rank elastic material tensor defined by
 \begin{equation}\label{eq:ten}
 \mathrm{C}_{ijkl}(\bx):=\lambda(\bx)\bm{\delta}_{ij}\bm{\delta}_{kl}+\mu(\bx)(\bm{\delta}_{ik}\bm{\delta}_{jl}+\bm{\delta}_{il}\bm{\delta}_{jk}),\ \ \bx\in\mathbb{R}^N,
 \end{equation}
 where $\bm{\delta}$ is the Kronecker delta. In \eqref{eq:ten}, $\lambda$ and $\mu$ are two scalar functions and referred to as the Lam\'e parameters. For a regular elastic material, the Lam\'e parameters satisfy the following strong convexity conditions,
\begin{equation}\label{eq:con}
  \mathrm{i)} ~~~\mu>0\,; \qquad \mathrm{ii)} ~~~ 2\lambda+N\mu>0.
 \end{equation}

Next, we introduce a core-shell-matrix material structure for our study. 
{Let $D$ and $\Omega\subset\mathbb{R}^N$ be two bounded $C^{1,\alpha}$-domains for $\alpha\in(0,1)$ 
such that $D\Subset\Omega$, and both $\Omega\backslash\overline{D}$ and 
$\mathbb{R}^N\backslash\overline{\Omega}$ are connected.} Assume that the matrix $\mathbb{R}^N\backslash\overline{\Omega}$ is occupied by a regular elastic material parameterized by two Lam\'e constants $(\lambda,\mu)$ satisfying \eqref{eq:con}. The shell $\Omega\backslash\overline{D}$ is occupied by a metamaterial whose Lam\'e constants are given by $(\hat{\lambda}, \hat{\mu})$. It is assumed that $\hat\lambda$ and $\hat\mu$ can be flexibly chosen and do not necessarily fulfil the strong convexity conditions \eqref{eq:con}. In fact, they are complex-valued with $\Re\hat{\lambda}, \Re\hat{\mu}$ breaking the strong convexity conditions \eqref{eq:con} and $\Im\hat{\lambda}, \Im\hat{\mu}\in\mathbb{R}_+$. This is critical in our study and shall be further remarked in what follows. The inner core $D$ is occupied by a regular elastic material whose Lam\'e constants $(\breve{\lambda}, \breve{\mu})$ satisfy the strong convex conditions \eqref{eq:con}. We write $\mathbf{C}_{\mathbb{R}^N\backslash\overline{\Omega},\lambda,\mu}$ to specify the dependence of the elastic tensor on the domain $\mathbb{R}^N\backslash\overline{\Omega}$ as well as the Lam\'e parameters $(\lambda,\mu)$. The same notation applies to the tensors $\mathbf{C}_{\Omega\backslash\overline{D},\hat{\lambda},\hat{\mu}}$ and $\mathbf{C}_{D,\breve{\lambda},\breve{\mu}}$. Now we introduce the following elastic tensor:
\begin{equation}\label{eq:pa1}
 \mathbf{C}_0=\mathbf{C}_{\mathbb{R}^N\backslash\overline{\Omega},\lambda,\mu} + \mathbf{C}_{\Omega\backslash\overline{D},\hat{\lambda},\hat{\mu}} + \mathbf{C}_{D,\breve{\lambda},\breve{\mu}}.
\end{equation}
The tensor $\mathbf{C}_0$ describes an elastic material configuration of a core-shell-matrix structure with the metamaterial located in the shell. We point out that it may happen that $D=\emptyset$ in our subsequent analysis. In such a case, $\mathbf{C}_0$ is said to be a metamaterial structure with no core. In what follows, material structures with a core or without a core can induce different resonance phenomena.

Let $\bff\in L^\infty_{loc}(\mathbb{R}^N\backslash\overline{\Omega})^N$ signify an elastic source that is compactly supported in $\mathbb{R}^N\backslash\overline{\Omega}$. The elastic displacement field $\bu=(u_i)_{i=1}^N\in H^1_{loc}(\mathbb{R}^N)^N$ induced by the interaction between the source $\bff$ and the medium configuration $\mathbf{C}_0$ is governed by the following Lam\'e system: 
\begin{equation}\label{eq:lame1}
\begin{cases}
& \nabla\cdot\mathbf{C}_0\nabla^s\mathbf{u}(\mathbf{x})+\omega^2 \mathbf{u}(\bx)=\mathbf{f}(\bx)\quad\mbox{in}\ \ \mathbb{R}^N,\medskip\\
& \mbox{$\mathbf{u}(\bx)$ satisfies the Kupradze radiation condition,}
\end{cases}
\end{equation}
where $\omega\in\mathbb{R}_+$ signifies an angular frequency. Here and also in what follows, the operator $\nabla^s$ is the symmetric gradient defined by
 \begin{equation}\label{eq:sg1}
 \nabla^s\mathbf{u}:=\frac{1}{2}\left(\nabla\mathbf{u}+\nabla\mathbf{u}^t \right),
 \end{equation}
 where $\nabla\bu$ denotes the matrix $(\partial_j u_i)_{i,j=1}^N$ and the superscript $t$ signifies the matrix transpose. It follows from \cite{Kup} that the elastic displacement $\mathbf{u}(\bx)$ can be decomposed into
 $ \bu =\bu_{p} + \bu_{s}$ {in $\mathbb{R}^N\backslash\overline\Omega$}, where $\bu_{p}$ and $\bu_{s}$ are respectively referred to as the pressure and shear waves and satisfy the following equations:
 \begin{equation}
     \left(\triangle+k_{p}^{2}\right) \mathbf{u}_{p}=0, \quad \nabla \times \mathbf{u}_{p}=0;\ \
      \left(\triangle+k_{s}^{2}\right) \mathbf{u}_{s}=0, \quad \nabla \cdot \mathbf{u}_{s}=0 ,
\end{equation}
 with 
 \begin{equation}\label{pa:ksp}
 k_{s}:=\omega / \sqrt{\mu} \quad \mbox{and}\quad k_{p}:=\omega / \sqrt{\lambda+2 \mu}.
\end{equation}
In \eqref{eq:lame1}, the the Kupradze radiation condition {is expressed as} 
 \begin{equation}\label{eq:radi}
\nabla \mathbf{u}_{p} \hat{\mathbf{x}}-\mathrm{i} k_{p} \mathbf{u}_{p} =\mathcal{O}\left(|\mathbf{x}|^{-(N+1)/2}\right)\quad\mbox{and}\quad
\nabla \mathbf{u}_{s} \hat{\mathbf{x}}-\mathrm{i} k_{s} \mathbf{u}_{s} =\mathcal{O}\left(|\mathbf{x}|^{-(N+1)/2}\right)
\end{equation}
{as $\abs{\bx}\rightarrow \infty$, which hold uniformly in $\hat{\bx}=\bx/|\bx|\in\mathbb{S}^{N-1}$.}

Next, we recall the quasi-static condition for the above elastic scattering problems:
\begin{equation}\label{eq:qs1}
\omega\cdot\mathrm{diam}(\Omega)\ll 1,
\end{equation}
which signifies that the size of the material structure $\Omega$, i.e. the diameter of $\Omega$, is much smaller than the operating wavelength $2\pi/\omega$. In the current article, we shall instead mainly study the case beyond the quasi-static regime, namely 
\begin{equation}\label{eq:qs2}
\omega\cdot\mathrm{diam}(\Omega)\sim 1. 
\end{equation}
For simplicity, it is sufficient for us to require that 
\begin{equation}\label{eq:qs3}
\omega\sim 1\quad\mbox{and}\quad \mathrm{diam}(\Omega)\sim 1. 
\end{equation}

We proceed to introduce the following functional for $\bu,\bv\in \big(H^1(\Omega\backslash\overline{D})\big)^N$:
\begin{equation}\label{eq:func1}
\begin{split}
 P_{\hat{\lambda},\hat{\mu}}(\bu,\bv) = &\int_{\Omega\backslash\overline{D}} \nabla^s\bu:\mathbf{C}_0\overline{\nabla^s\bv}d\bx 
 =  \int_{\Omega\backslash \overline{D}}\Big(\hat{\lambda}(\nabla\cdot\mathbf{\bu})\overline{(\nabla\cdot\mathbf{\bv})} + 2\hat{\mu}\nabla^s\mathbf{u}:\overline{\nabla^s\mathbf{\bv}} \Big)\ d \bx,
 \end{split}
\end{equation}
where  $\mathbf{C}_0$ and $\nabla^s$ are defined in
\eqref{eq:pa1} and \eqref{eq:sg1}, respectively. In \eqref{eq:func1} and also in what follows, $\mathbf{A}:\mathbf{B}=\sum_{i,j=1}^N a_{ij}b_{ij}$ for two matrices $\mathbf{A}=(a_{ij})_{i,j=1}^N$ and $\mathbf{B}=(b_{ij})_{i,j=1}^N$. 
The energy dissipation of the elastic system \eqref{eq:lame1}--\eqref{eq:radi} is given by 
\begin{equation}\label{def:E}
\mathscr{E}(\bu)=\Im P_{\hat{\lambda},\hat{\mu}}(\bu,\bu).
\end{equation}
We are now in a position to give the precise meaning of the elastic resonances for our subsequent study. 
\begin{defn}\label{def:1}
Consider the Lam\'e system \eqref{eq:lame1}--\eqref{eq:radi} associated with the material structure $\mathbf{C}_0$ 
in
\eqref{eq:pa1} under the assumption \eqref{eq:qs3}. We say that resonance occurs if it holds that
\begin{equation}\label{con:res}
 \mathscr{E}(\bu) \geq M 
\end{equation}
for $M\gg 1$. 
If in addition to \eqref{con:res}, the displacement field $\bu$ further fulfils the following boundedness condition:
\begin{equation}\label{con:bou}
|\bu|\leq C\quad \mbox{when} \quad |\bx|>\tilde{R},
\end{equation}
for some $C\in\mathbb{R}_+$ and $\tilde{R}\in\mathbb{R}_+$ such that $\Omega\subset B_{\tilde{R}}$, then we say that anomalous localized resonance (ALR) occurs. Here and also in what follows, $B_{\tilde R}$ signifies a ball of radius $\tilde R$ and centred at the origin, and $C, \tilde{R}$ are constants independent of $\mathbf{C}_0$ and $\mathbf{f}$. 
\end{defn}

\begin{rem}\label{rem:s1}
It is noted that the resonant condition \eqref{con:res} indicates that the resonant field $\mathbf{u}$ exhibits highly oscillatory behaviour.  Moreover, in our subsequent study, it allows that $M\rightarrow+\infty$, which indicates that in the limiting case, the scattering system \eqref{eq:lame1}--\eqref{eq:radi} loses its well-posedness. Indeed, it shall be seen in what follows that in the limiting case, the solutions to the scattering system \eqref{eq:lame1}--\eqref{eq:radi} 
are not unique. It is clear that the metamaterials located in $\Omega\backslash\overline{D}$ plays a critical role for the occurrence of the resonance. In fact, if $\mathbf{C}_0$ is a regular elastic material configuration, then the Lam\'e system \eqref{eq:lame1}--\eqref{eq:radi} is well-posed and the resonance does not occur.
\end{rem}

\begin{rem}\label{rem:s2}
If ALR occurs, one can show invisibility cloaking phenomenon can be induced. In fact, by normalization, we set $\widetilde{\mathbf{f}}:=\mathbf{f}/\sqrt{\bm{\alpha}}$, where $\bm{\alpha}:=\Im P_{\hat{\lambda},\hat{\mu}}(\bu,\bu)\gg 1$. One can see that both $\widetilde{\mathbf{f}}$ and the material structure $\mathbf{C}_0$ are nearly invisible to observations made outside $B_{\tilde{R}}$. Indeed, it is easily seen that the induced elastic field $\widetilde{\mathbf{u}}=\mathbf{u}/\sqrt{\bm{\alpha}}\ll 1$ in $\mathbb{R}^N\backslash B_{\tilde{R}}$; see \cite{Ack13,DLL2,LiLiu2d,LLL1} for more relevant discussions. Hence, when ALR occurs according to Definition~\ref{def:1}, we also say that cloaking due to anomalous localized resonance (CALR) occurs.   
\end{rem}

The major findings of this article can be briefly summarized as follows with the technical details supplied in 
the sequel; see Theorems \ref{th:reg}, \ref{thm:reson} and \ref{thm:CALR}:

\medskip
Consider the Lam\'e system \eqref{eq:lame1}--\eqref{eq:radi} associated with the material structure $\mathbf{C}_0$ 
in \eqref{eq:pa1}, under the assumption \eqref{eq:qs3}.
\begin{enumerate}
\item Suppose that the material structure has no core, namely $D=\emptyset$. There exist generic material structures of the form \eqref{eq:pa1} such that resonance occurs. 

\item Under the same setup as the above (1), but with $\Omega$ being radially symmetric, we derive the explicit construction of all the material structures that can induce resonance. Moreover, we present a comprehensive analysis on the quantitative behaviours of the resonant field. It is shown that the resonance behaviours possess distinct characters, with some similar to the subwavelength resonances, but some sharply different due to the frequency effect.

\item We construct a core-shell metamaterial structure that can induce CALR beyond the quasi-static limit. 

\item In establishing the resonance results, we derive comprehensive spectral properties of the non-static elastic Neumann-Poincar\'e (N-P) operator, {which will be introduced in the sequel. 
These results} are of independent interest to the spectral theory of layer potential operators. 
\end{enumerate}

%
%
%
%

\begin{rem}\label{rem:s2}
According to our discussion in Remark~\ref{rem:s1}, the main technical ingredient in our study is to derive some relations satisfied by the material parameters in $\mathbf{C}_0$, the geometric parameters of $\Omega/D$ as well as the frequency $\omega$ such that the resonance conditions (Definition~\ref{def:1}) can be fulfilled. It is clear that these conditions are coupled nonlinearly and in fact they are essentially determined by the infinite-dimensional kernel of the PDE system \eqref{eq:lame1}, i.e. the set of nontrivial solutions to \eqref{eq:lame1} with $\mathbf{f}\equiv\mathbf{0}$. 
\end{rem}

\begin{rem}\label{rem:s3}
We would like to make a remark on the metamaterial parameters in $\Omega\backslash\overline{D}$, namely $\hat\lambda$ and $\hat\mu$. As pointed out in Remark~\ref{rem:s1}, $\Re\hat\lambda$ and $\Re\hat\mu$ are allowed to break the strong convexity conditions in \eqref{eq:con}. This is critical for inducing the resonances. On the other hand, $\Im\hat\lambda$ and $\Im\hat\mu$ are required to be positive. In a certain sense, they play the role of regularization parameters that can retain the well-posedness of the Lam\'e system \eqref{eq:lame1}--\eqref{eq:radi}. Moreover, they are also critical physical parameters in order to induce the resonances. In fact, $\Im\hat\lambda$ and $\Im\hat\mu$ should be delicately chosen according to $\Re\hat\lambda$, $\Re\hat\mu$ and $\omega$ as well as the asymptotic parameter $M$ in \eqref{con:res}. This is in sharp contrast to the related studies in the static/quasi-static case where $\Im\hat\lambda$ and $\Im\hat\mu$ play solely as the regularization parameters which are asymptotically small generic parameters and converge to zero in the limiting case. This shall become clearer in our subsequent analysis. 
\end{rem}


\subsection{Connection to existing studies and discussions}

Metamaterials are artificially engineered materials to have properties that are not found in naturally occurring materials. Negative materials are an important class of metamaterials which possess negative material parameters. Negative materials can be artificially engineered by assembling subwavelength resonators periodically or randomly; see e.g. \cite{AZ, LLZ,SPVNS} and the references cited therein. Negative materials are revolutionizing many industrial applications including antennas \cite{ETS}, absorber \cite{LV}, invisibility cloaking \cite{Ack13, s25, LLL,GWM3,GWM7}, superlens \cite{FLS, P} and super-resolution imaging \cite{ARYZ,DLZ3}, to mention just a few. 

We are mainly concerned with the quantitative theoretical understandings of negative metamaterials, which have received considerable interest recently in the literature. A variety of peculiar resonance phenomena form the fundamental basis for many striking applications of negative metamaterials. Intriguingly, those resonance phenomena are distinct and possess distinguishing characters. For a typical scenario, let us consider the Lam\'e system \eqref{eq:lame1} in the static case, namely $\omega\equiv 0$. If $\mathbf{C}_0$ is allowed to possess negative material parameters, it is no longer an elliptic tensor, i.e. the strong convexity conditions \eqref{eq:con} can be broken. In such a case, the PDE system \eqref{eq:lame1} may possess (infinitely many) nontrivial solutions even with $\mathbf{f}\equiv\mathbf{0}$. Hence, the infinite kernel of the non-elliptic partial differential operator (PDO), namely $\nabla\cdot\mathbf{C}_0\nabla^s$, can induce certain resonances if the excitation term $\mathbf{f}$ is properly chosen. Similar resonance phenomena have been more extensively and intensively investigated for acoustic and electromagnetic metamaterials that are governed by the Helmholtz and Maxwell systems, respectively. They are referred to as the plasmon/polariton resonances in the literature; see \cite{Ack13, AKL, BLL,DLZ1,Z} for the Helmholtz equation, \cite{ARYZ,DLZ2, LLLW} for the Maxwell system, and \cite{DLL,DLL1, DLL2, LiLiu3d, LLL1} for the Lam\'e system. Most of the existing studies in the literature are concerned with the static or quasi-static cases (cf. \eqref{eq:qs1}). A widely studied resonance phenomenon is induced by the interface of negative and positive materials, which is referred to as the plasmon/polariton resonance in the literature. It turns out that the plasmon/polariton resonant oscillations are localized around the metamaterial interface, and hence are usually called the surface plasmon/polariton resonances.  

It is not surprising that the occurrence of plasmon/polariton resonances strongly depend on the medium configuration as well as the geometry of the metamaterial structure, which are delicately coupled together in certain nonlinear relations. In this paper, we {for the first time} show the existence of generic metamaterial structures that can induce resonances in elasticity beyond quasi-static approximations {in both 2D and 3D}. It turns out that in addition to the medium and geometric parameters of the metamaterial structure, the operating frequency shall also play a critical role and needs to be incorporated into the nonlinear coupling mentioned above. In addition to its theoretical significance, we would like to emphasize that our study also uncover two interesting physical phenomena due to the frequency effect. First, the resonant oscillation outside the material structure is localized around the metamaterial interface, but inside the material structure it is not localized around the interface, which is sharply different from the subwavelength resonances; see more detailed discussion at the end of Section 3. Second, as already commented in Remark~\ref{rem:s3}, the loss parameters $\Im\hat\lambda$ and $\Im\hat\mu$ shall also play an important role, and they generally are required to be non-zero constants in the limiting case; see Remark~\ref{rem:nn1} in what follows for more details. Finally, as noted earlier, negative materials usually occur in the nanoscale, and hence it is unobjectionable that many studies are concerned with subwavelength resonances. On the other hand, there are also conceptual and visionary studies which employ metamaterials for novel applications beyond the quasi-static limit, say e.g. the superlens \cite{FLS, P}. The proposed study in this paper follows a similar spirit to the latter class mentioned above, though we are mainly concerned with the theoretical aspects. 

If the metamaterial structure is constructed in the core-shell form, it may induce the cloaking phenomenon due to the anomalous localized resonance \cite{bk:ab, DLL, LiLiu2d}; that is the whole structure is invisible for an impinging wave. This is a much more delicate and subtle resonance phenomenon: the resonant oscillation is localized within a bounded region, i.e. $B_{\tilde R}$ in Definition~\ref{def:1}, and moreover it not only depends on the material and geometric configurations of the core-shell structure, but also critically depends on the location of the excitation source. In addition to the invisibility phenomenon mentioned above, it is observed in that any small objects located near the material structure within the critical radius are also invisible to faraway observations; see \cite{GWM3,GWM7} for related discussions. All of the existing studies are confined within the radial geometries since on the one hand, one needs explicit expressions of the spectral system of certain integral operators \cite{Ack13,DLL,LiLiu2d,s25,LLLW,DLL} and on the other hand it seems unnecessary for constructing material structures of general shapes for the cloaking purpose. 
The CALR was recently studied in \cite{DLL2} for 3D elasticity beyond the quasi-static approximation for the spherical structure. Hence, in the current article, we mainly consider the CALR in two dimensions. Nevertheless, we would like to remark that the derivation in 2D is more subtle and technically involved. The main reason is that in 3D elasticity there exists a certain class of shear waves that can be decoupled from the other shear waves and compressional waves \cite{DLL2}. The decoupling property significantly simplifies the analysis of the CALR. However, in 2D elasticity all the shear waves and compressional waves are coupled together, which substantially increase the complexity of the relevant theoretical analysis. In fact, we develop several technically new ingredients in handling the 2D case in the present article.

Finally, we would like to discuss one more technical novelty of our study. In studying the metamaterial resonances, one powerful tool is to make use of the layer potential theory to reduce the underlying PDE system into a system of certain integral operators. In doing so, the resonance analysis is boiled down to analyzing the spectral properties of the integral operators. In this paper, we provide an in-depth analysis of the so-called elastic Neumann-Poincar\'e (N-P) operator in the frequency regime. In particular, we derive the complete spectral system of the elastic N-P operator with several interesting observations. These results are new to the literature and are of independent interest to the spectral theory of elastic layer potential operators (cf. \cite{AKM,DLL,MR,R}) .

The rest of the paper is organized as follows. In Section 2, we present several technical auxiliary results. Section 3 is devoted to resonance analysis for material structures with no core. In Section 4, we construct a core-shell structure that can induce cloaking due to anomalous localized resonance. 

\section{Auxiliary results}\label{sec:auxiliary} 

In this section, we derive some key auxiliary results that will be needed for our subsequent analysis. 

Set  $\bx=(x_j)_{j=1}^N\in\mathbb{R}^N$ to be the Euclidean coordinates and $r=\abs{\bx}$. Let $\theta_{\bx}$ be the angle between $\bx$ and $x_1$-axis. If there is no ambiguity, we write $\theta$ instead of $\theta_{\bx}$ for simplicity. Let $\bnu$ signify the outward unit normal to a boundary $\partial \Omega$. If the domain $\Omega$ is a circle $B_R$, then $\bnu = (\cos(\theta), \sin(\theta))^t$ and the direction $\mathbf{t} = (-\sin(\theta), \cos(\theta))^t$ is the tangential direction on the boundary $\partial B_R$. Denote by $\nabla_{\bbS}$ the surface gradient.

The Lam\'e operator $\Lcal_{\lambda,\mu}$ associated with the parameters $(\lambda, \mu)$ is defined by
\begin{equation}\label{op:lame}
 \Lcal_{\lambda,\mu}\bw:=\mu \triangle\bw + (\lambda+ \mu)\nabla\nabla\cdot\bw.
\end{equation}
The traction (the conormal derivative) of $\bw$ on the boundary $\partial \Omega$ is defined by
\begin{equation}\label{eq:trac}
\partial_{\bnu}\bw=\lambda(\nabla\cdot \bw)\bnu + 2\mu(\nabla^s\bw) \bnu,
\end{equation}
where the operator $\nabla^s$ is defined in \eqref{eq:sg1}.
From \cite{bk:ab}, the fundamental solution $\bGa^{\omega}=(\Gamma^{\omega}_{i,j})_{i,j=1}^2$ to the operator $\Lcal_{\lambda,\mu}+\omega^2$ in two dimensions is given by
\begin{equation}\label{eq:ef}
 \left(\Gamma_{i, j}^{\omega}\right)_{i, j=1}^{2}(\mathbf{x})=-\frac{\rmi \bm{\delta}_{i j}}{4 \mu} H_{0}\left(k_{s}|\mathbf{x}|\right)+\frac{\mathrm{i}}{4 \pi \omega^{2}} \partial_{i} \partial_{j}\left(H_{0}\left(k_{p}|\mathbf{x}|\right)-H_{0}\left(k_{s}|\mathbf{x}|\right)\right),
\end{equation}
where $H_0(\cdot)$ is the Hankel function of the first kind of order 0, and $k_s$ and $k_p$ are defined in \eqref{pa:ksp}. 
The corresponding fundamental solution $\bGa^{\omega}=(\Gamma^{\omega}_{i,j})_{i,j=1}^3$ in three dimensions are given by
\begin{equation}\label{eq:ef3}
 (\Gamma^{\omega}_{i,j})_{i,j=1}^3(\bx)=-\frac{\bm{\delta}_{ij}}{4\pi\mu|\bx|}e^{\rmi k_s |\bx|} + \frac{1}{4\pi \omega^2}\partial_i\partial_j\frac{e^{\rmi k_p|\bx|} - e^{\rmi k_s|\bx|}}{|\bx|}.
\end{equation}
Then the single-layer potential associated with the fundamental solution $\bGa^{\omega}$ is defined as 
\begin{equation}\label{eq:single}
 \bS_{\partial\Omega}^{\omega}[\bvarphi](\bx)=\int_{\partial \Omega} \bGa^{\omega}(\bx-\by)\bvarphi(\by)ds(\by), \quad \bx\in\mathbb{R}^N,
\end{equation}
for $\bvarphi\in L^2(\partial \Omega)^N$. On the boundary $\partial \Omega$, the conormal derivative of the single-layer potential satisfies the following jump formula
\begin{equation}\label{eq:jump}
 \frac{\partial \bS_{\partial\Omega}^{\omega}[\bvarphi]}{\partial \bnu}|_{\pm}(\bx)=\left( \pm\frac{1}{2}\bI + \bK_{\partial\Omega}^{\omega,*} \right)[\bvarphi](\bx) \quad \bx\in\partial \Omega,
\end{equation}
where
\[
\bK_{\partial\Omega}^{\omega,*}[\bvarphi](\bx)=\mbox{p.v.} \int_{\partial \Omega} \frac{\partial \bGa^{\omega}}{\partial \bnu(\bx)}(\bx-\by)\bvarphi(\by)ds(\by),
\]
with $\mbox{p.v.}$ standing for the Cauchy principal value and the subscript $\pm$ indicating the limits from outside and inside $\Omega$, respectively. The operator $\bK_{\partial\Omega}^{\omega,*}$ is called the Neumann-Poincar\'e (N-P) operator associated with the Lam\'e system.

Next, we present some properties of the N-P operator $\bK_{\partial\Omega}^{\omega,*}$. It is shown in \cite{AKM} that the operator $\bK_{\partial\Omega}^{\omega,*}$ is not compact and only polynomially compact in the following sense.
\begin{lem}
The N-P operator $\bK_{\partial\Omega}^{\omega,*}$ is polynomially compact in the sense that in two dimensions, the operator $ (\bK_{\partial\Omega}^{\omega,*})^2 - k_0^2 \mathbf{I}$ is compact; 
while in three dimensions, the operator $\bK_{\partial\Omega}^{\omega,*} \left( (\bK_{\partial\Omega}^{\omega,*})^2 - k_0^2 \mathbf{I} \right)$ is compact, where 
\[
 k_0:=\frac{\mu}{2(\lambda + 2\mu)}.
\]
\end{lem}

Then we can derive the following lemma for the spectrum of the N-P operator $\bK_{\partial\Omega}^{\omega,*}$ (cf. \cite{K}).
\begin{lem}\label{lem:sks}
In two dimensions, the spectrum $\sigma(\bK_{\partial\Omega}^{\omega,*})$ consists of two nonempty sequences of eigenvalues that converge to $k_0$ and $-k_0$, respectively; while in three dimensions, the spectrum $\sigma(\bK_{\partial\Omega}^{\omega,*})$ consists of three nonempty sequences of eigenvalues that converge to $0$, $k_0$ and $-k_0$, respectively.
\end{lem}

Let $ \Phi(\bx)$ be the fundamental solution to the operator $\triangle + k^2$ in two dimensions given by
\begin{equation}\label{eq:ful}
  \Phi(\bx)=-\frac{\rmi}{4}H_0(k\abs{\bx}).
\end{equation}
For $\varphi\in L^2(\partial\Omega)$, we define the single-layer potential associated with the fundamental solution $ \Phi(\bx)$ by
\begin{equation}\label{eq:sh}
 S_{\partial\Omega}^{k}[\varphi](\bx)=\int_{\partial \Omega} \Phi(\bx-\by)\varphi(\by)ds(\by), \quad \bx\in\mathbb{R}^2.
\end{equation}

Let $J_n(t)$ and $H_n(t)$, $n\in\mathbb{Z}$, denote the Bessel and Hankel functions of the first kind of order $n$, respectively. These functions satisfy the following Bessel differential equation
\begin{equation}\label{eq:bd}
t^{2} f^{\prime \prime}(t)+t f^{\prime}(t)+\left(t^{2}-n^{2}\right) f(t)=0,
\end{equation}
for $f=J_n$ or $H_n$. When $n$ is negative, there hold that $J_n(t)=(-1)^nJ_{-n}(t)$ and $H_n(t)=(-1)^nH_{-n}(t)$. Moreover, the Bessel and Hankel functions $J_n(t)$ and $H_n(t)$ satisfy the recursion formulas (cf. \cite{CK}):
\begin{equation}\label{eq:re}
\begin{split}
f_{n+1}^{\prime}=f_{n}-(n+1) \frac{f_{n+1}}{t},  \quad & f_{n+1}=n \frac{f_{n}}{t}-f_{n}^{\prime}, \quad \text { for } \quad n \geq 0, \\
f_{n-1}^{\prime}=-f_{n}+(n-1) \frac{f_{n-1}}{t}, \quad & f_{n-1}=n \frac{f_{n}}{t}+f_{n}^{\prime}, \quad \text { for } \quad n \geq 1.
\end{split}
\end{equation}
The following asymptotic expansions hold for $t\ll1$ (cf. \cite{CK}),
\begin{equation}\label{eq:asj}
\begin{aligned}
J_{n}(t) &=\frac{t^{n}}{2^{n} n !}\left(1-\frac{t^{2}}{4(n+1)}+o\left(t^{2}\right)\right), & \text { for } & n \geq 0, \\
H_{n}(t) &=\frac{-\mathrm{i} 2^{n} n !}{\pi t^{n}}\left(1+\frac{t^{2}}{4(n-1)}+o\left(t^{2}\right)\right), & \text { for } & n \geq 2,
\end{aligned}
\end{equation}
and
\begin{equation}\label{eq:ash1}
H_{1}(t) =-\frac{\mathrm{i} 2}{\pi t}\left(1-\frac{t^{2}}{2}\left(\ln t-\ln 2+E u-\frac{1+\mathrm{i} \pi}{2}\right)+o\left(t^{3}\right)\right) ,
\end{equation}
where 
$Eu$ is the Euler constant. For larger $n$, the following asymptotic expansions hold
\begin{equation*}\label{eq:asn}
\begin{split}
  J_n(t)=  \frac{(t/2)^n}{n!}\Big\{1- \frac{t^2}{4n} + \frac{8t^2 + t^4}{32n^2} + o\Big\{\frac{1}{n^2} \Big\}\Big\},\ \
  H_n(t)=  \frac{-\rmi n!}{\pi(t/2)^n} \Big\{ \frac{1}{n} + \frac{t^2}{n^2} + o\Big\{\frac{1}{n^2} \Big\}  \Big\}.
 \end{split}
\end{equation*}
%
%
The Hankel function $H_{0}(t) $ has the following expansion from Graf’s formula(cf. \cite{Fdc})
\begin{equation}\label{eq:expH}
H_{0}(k|\mathbf{x}-\mathbf{y}|)=\sum_{n \in \mathbb{Z}} H_{n}(k|\mathbf{x}|) e^{\mathrm{i} n \theta \mathbf{x}} J_{n}(k|\mathbf{y}|) e^{-\mathrm{i} n \theta \mathbf{y}}.
\end{equation}

%

We will also need the following single-layer potential $S_{\partial B_{R}}^{k}$ acting on the density $e^{\mathrm{i} n \theta}$.
\begin{lem}
Let $S_{\partial B_{R}}^{k}$ be defined in \eqref{eq:sh}, then it holds that
\begin{equation*}
S_{\partial B_{R}}^{k}[e^{\mathrm{i} n \theta}](\mathbf{x})=-\frac{\mathrm{i} \pi R}{2} J_{n}(k R) H_{n}(k|\mathbf{x}|) e^{\mathrm{i} n \theta}, \quad \forall\,\bx\in\mathbb{R}^2\backslash\overline{B}_R\,.
\end{equation*}
\end{lem}
\begin{proof}
 From the definition of the operator $S_{\partial\Omega}^{k}$ in \eqref{eq:sh} and the expansion for the function $H_{0}(k|\mathbf{x}-\mathbf{y}|)$ in \eqref{eq:expH}, one has that
\begin{equation*}
\begin{aligned}
& S_{\partial B_{R}}^{\omega}[e^{i n \theta}](\mathbf{x}) =-\frac{\mathrm{i}}{4} \int_{\partial B_{R}} \sum_{m \in \mathbb{Z}} H_{m}(k|\mathbf{x}|) e^{\mathrm{i} m \theta_{\mathbf{x}}} J_{m}(k|\mathbf{y}|) e^{-\mathrm{i} m \theta_{\mathbf{y}}} e^{\mathrm{i} n \theta_{\mathbf{y}}} d s_{\by} \\
&=-\frac{\mathrm{i} R}{4} \int_{0}^{2 \pi} \sum_{m \in \mathbb{Z}} H_{m}(k|\mathbf{x}|) e^{\mathrm{i} m \theta_{\mathbf{x}}} J_{m}(k R) e^{\mathrm{i}(n-m) \theta_{\mathbf{y}}} ds_{\by} =-\frac{\mathrm{i} \pi R}{2} J_{n}(k R) H_{n}(k|\mathbf{x}|) e^{\mathrm{i} n \theta}.
\end{aligned}
\end{equation*}
This completes the proof.
\end{proof}

For further calculations, we need the following identities.
\begin{lem}\label{lem:in}
There hold that: if $m = n-1$,
\[
\int_{\bbS} e^{-\rmi m \theta} e^{\rmi n \theta} \bnu d s=\pi\left(\begin{array}{c}
1 \\
\rmi
\end{array}\right), \quad \int_{\bbS} \nabla_{\bbS} e^{-\rmi m \theta} e^{\rmi n \theta} d s=-(n-1) \pi\left(\begin{array}{c}
1 \\
\rmi
\end{array}\right);
\]
and if $m = n+1$,
\[
\int_{\bbS} e^{-\rmi m \theta} e^{\rmi n \theta} \bnu d s=\pi\left(\begin{array}{c}
1 \\
-\rmi
\end{array}\right), \quad \int_{\bbS} \nabla_{\bbS} e^{-\rmi m \theta} e^{\rmi n \theta} d s=(n+1) \pi\left(\begin{array}{c}
1 \\
-\rmi
\end{array}\right).
\]
\end{lem}

\begin{proof}
Direct calculations yield that
\begin{equation*}
\int_{\mathbb{S}} e^{-\mathrm{i} m \theta} e^{\rmi n \theta} \bnu d s=\int_{\mathbb{S}} e^{\rmi(n-m) \theta}\left(\begin{array}{c}
\cos (\theta) \\
\sin (\theta)
\end{array}\right) d s=\int_{\bbS} e^{\rmi(n-m) \theta}\left(\begin{array}{c}
\frac{e^{\rmi \theta}+e^{-\rmi \theta}}{2} \\
\frac{e^{\rmi \theta}-e^{-\rmi \theta}}{2}
\end{array}\right) d s.
\end{equation*}
Thus if $m = n - 1$, one has that	
\[
\int_{\bbS} e^{-\rmi m \theta} e^{\rmi n \theta} \bnu d s=\pi\left(\begin{array}{c}
1 \\
\rmi
\end{array}\right);
\]
and if $m = n + 1$, one has that
\[
\int_{\bbS} e^{-\rmi m \theta} e^{\rmi n \theta} \bnu d s=\pi\left(\begin{array}{c}
1 \\
-\rmi
\end{array}\right).
\]
Furthermore, one can have that
\begin{equation*}
\int_{\mathbb{S}} \nabla_{\bbS} e^{-\mathrm{i} m \theta} e^{\rmi n \theta} \bnu d s=-\rmi m\int_{\mathbb{S}} e^{\rmi(n-m) \theta}\left(\begin{array}{c}
-\sin (\theta) \\
\cos (\theta)
\end{array}\right) d s=-\rmi m \int_{\bbS} e^{\rmi(n-m) \theta}\left(\begin{array}{c}
\frac{e^{\rmi \theta}-e^{-\rmi \theta}}{-2\rmi} \\
\frac{e^{\rmi \theta}+e^{-\rmi \theta}}{2}
\end{array}\right) d s.
\end{equation*}
Thus if $m = n - 1$, one has that	
\[
 \int_{\bbS} \nabla_{\bbS} e^{-\rmi m \theta} e^{\rmi n \theta} d s=-(n-1) \pi\left(\begin{array}{c}
1 \\
\rmi
\end{array}\right);
\]
and if $m = n + 1$, one has that
\[
\int_{\bbS} \nabla_{\bbS} e^{-\rmi m \theta} e^{\rmi n \theta} d s=(n+1) \pi\left(\begin{array}{c}
1 \\
-\rmi
\end{array}\right).
\]
This completes the proof.
\end{proof}

\begin{lem}\label{lem:d1}
The following two identities hold for $\bx\in\mathbb{R}^2\backslash\overline{B}_R$, 
\begin{equation*}
\begin{aligned}
&\int_{\partial B_{R}} H_{0}(k|\mathbf{x}-\mathbf{y}|) e^{\mathrm{i} n \theta_{\mathbf{y}}} \bnu_{\mathbf{y}} d s_{\by}\\
=&\pi R H_{n-1}(k|\mathbf{x}|) e^{\mathrm{i}(n-1) \theta} J_{n-1}(k R) &\left(\begin{array}{c}
1 \\
\mathrm{i}
\end{array}\right)+
\pi R H_{n+1}(k|\mathbf{x}|) e^{\mathrm{i}(n+1) \theta} J_{n+1}(k R) &\left(\begin{array}{c}
1 \\
-\mathrm{i}
\end{array}\right),
\end{aligned}
\end{equation*}
and
\begin{equation*}
\begin{aligned}
& \int_{\partial B_{R}} H_{0}(k|\mathbf{x}-\mathbf{y}|) e^{\mathrm{i} n \theta_{\mathbf{y}}} \mathbf{t}_{\mathbf{y}} d s_{\by}\\
=&-\mathrm{i} \pi R H_{n-1}(k|\mathbf{x}|) e^{\mathrm{i}(n-1) \theta} J_{n-1}(k R)\left(\begin{array}{c}
1 \\
\mathrm{i}
\end{array}\right)+ \mathrm{i} \pi R H_{n+1}(k|\mathbf{x}|) e^{\mathrm{i}(n+1) \theta} J_{n+1}(k R)\left(\begin{array}{c}
1 \\
-\mathrm{i}
\end{array}\right).
\end{aligned}
\end{equation*}
\end{lem}
\begin{proof}
 From the expansion of the function $H_{0}(k|\mathbf{x}-\mathbf{y}|)$ in \eqref{eq:expH}, and with the help of Lemma \ref{lem:in}, one has that
\begin{equation*}
\begin{aligned}
& \int_{\partial B_{R}} H_{0}(k|\mathbf{x}-\mathbf{y}|) e^{\rmi n \theta_{\mathbf{y}}} \bnu_{\mathbf{y}} d s_{\by}=\int_{\partial B_{R}} \sum_{m \in \mathbb{Z}} H_{m}(k|\mathbf{x}|) e^{\mathrm{i} m \theta_{\mathbf{x}}} J_{m}(k|\mathbf{y}|) e^{-\mathrm{i} m \theta_{\mathbf{y}}} e^{\mathrm{i} n \theta_{\mathbf{y}}} \bnu_{\mathbf{y}} d s_{\by} \\
=& \sum_{m \in \mathbb{Z}} H_{m}(k|\mathbf{x}|) e^{\mathrm{i} m \theta_{\mathbf{x}}} R \int_{\mathbb{S}} J_{m}(k|\mathbf{y}|) e^{-\mathrm{i} m \theta_{\mathbf{y}}} e^{\mathrm{i} n \theta_{\mathbf{y}}} \bnu_{\mathbf{y}} d s_{\by} \\
=& \pi R H_{n-1}(k|\mathbf{x}|) e^{\mathrm{i}(n-1) \theta} J_{n-1}(k R)\left(\begin{array}{c}
1 \\
\mathrm{i}
\end{array}\right)+\pi R H_{n+1}(k|\mathbf{x}|) e^{\mathrm{i}(n+1) \theta} J_{n+1}(k R)\left(\begin{array}{c}
1 \\
-\mathrm{i}
\end{array}\right).
\end{aligned}
\end{equation*}
Thus the first identity is proved. Then, noticing that $\nabla_{\bbS} e^{\rmi n \theta} =\rmi n e^{\rmi n \theta}\mathbf{t}$, and following the same deduction one can obtain the other identity, which completes the proof.
\end{proof}

Based on the previous two lemmas, we can further have following identities.
\begin{lem}\label{lem:gnu}
There holds the following identity for $\bx\in\mathbb{R}^2\backslash\overline{B}_R$, 
\[
\begin{split}
& \int_{\partial B_{R}} \nabla_{\mathbf{x}} \nabla_{\mathbf{x}} H_{0}(k|\mathbf{x}-\mathbf{y}|) e^{\rmi n \theta_{\mathbf{y}}} \bnu_{\mathbf{y}} d s_{\by}= \\
&\frac{\pi H_{n-1}(k|\mathbf{x}|) e^{\rmi(n-1) \theta}}{R}\left(\left(n(n-1)-k^{2} R^{2}\right) J_{n-1}(k R)-n k R J_{n-1}^{\prime}(k R)\right)\left(\begin{array}{c}
1 \\
\mathrm{i}
\end{array}\right)+ \\
&\frac{\pi H_{n+1}(k|\mathbf{x}|) e^{\rmi(n+1) \theta}}{R}\left(\left(n(n+1)-k^{2} R^{2}\right) J_{n+1}(k R)+n k R J_{n+1}^{\prime}(k R)\right)\left(\begin{array}{c}
1 \\
-\mathrm{i}
\end{array}\right).
\end{split}
\]
\end{lem}

\begin{proof}
First note that $\nabla_{\mathbf{x}} \nabla_{\mathbf{x}} H_{0}(k|\mathbf{x}-\mathbf{y}|) = \nabla_{\mathbf{y}} \nabla_{\mathbf{y}} H_{0}(k|\mathbf{x}-\mathbf{y}|)$. Then one has that
\begin{equation}\label{eq:de2}
\begin{aligned}
& \int_{\partial B_{R}} \nabla_{\mathbf{x}} \nabla_{\mathbf{x}} H_{0}(k|\mathbf{x}-\mathbf{y}|) e^{\rmi n \theta_{\mathbf{y}}} \bnu_{\mathbf{y}} d s_{\by}=\int_{\partial B_{R}} \nabla_{\mathbf{y}} \nabla_{\mathbf{y}} H_{0}(k|\mathbf{x}-\mathbf{y}|) e^{\rmi n \theta_{\mathbf{y}}} \bnu_{\mathbf{y}} d s_{\by} \\
=& \int_{\partial B_{R}}\left(\nabla_{\mathbb{S}} \nabla_{\mathbf{y}} H_{0}(k|\mathbf{x}-\mathbf{y}|) /|\mathbf{y}|+\partial_{r} \nabla_{\mathbf{y}} H_{0}(k|\mathbf{x}-\mathbf{y}|) \bnu_{\mathbf{y}}\right) e^{\rmi n \theta_{\mathbf{y}}} \nu_{\mathbf{y}} d s_{\by} \\
=& \int_{\partial B_{R}}\left(\partial_{r} \nabla_{\mathbf{y}} H_{0}(k|\mathbf{x}-\mathbf{y}|)\right) e^{\rmi n \theta_{\mathbf{y}}} d s_{\by}.
\end{aligned}
\end{equation}
Furthermore, from the expansion of the function $H_{0}(k|\mathbf{x}-\mathbf{y}|)$ in \eqref{eq:expH}, one can have that
\begin{equation}\label{eq:der}
\begin{aligned}
& \partial_{r} \nabla_{\mathbf{y}} H_{0}(k|\mathbf{x}-\mathbf{y}|)=\partial_{r} \nabla_{\mathbf{y}}\left(\sum_{n \in \mathbb{Z}} H_{n}(k|\mathbf{x}|) e^{\rmi n \theta_{\mathbf{x}}} J_{n}(k|\mathbf{y}|) e^{-\rmi n \theta_{\mathbf{y}}}\right) \\
=& \sum_{n \in \mathbb{Z}} H_{n}(k|\mathbf{x}|) e^{\rmi n \theta_{\mathbf{x}}} \partial_{r}\left(k J_{n}^{\prime}(k|\mathbf{y}|) e^{-\mathrm{i} n \theta_{\mathbf{y}}} \bnu_{\mathbf{y}}+J_{n}(k|\mathbf{y}|) \nabla_{\mathbb{S}} e^{-\mathrm{i} n \theta_{\mathbf{y}}} /|\mathbf{y}|\right) \\
=& \sum_{n \in \mathbb{Z}} \frac{H_{n}(k|\mathbf{x}|) e^{\mathrm{i} n \theta_{\mathbf{x}}}}{|\mathbf{y}|^{2}}\left(k^{2}|\mathbf{y}|^{2} J_{n}^{\prime \prime}(k|\mathbf{y}|) e^{-\mathrm{i} n \theta_{\mathbf{y}}} \bnu_{\mathbf{y}}+\left(k|\mathbf{y}| J_{n}^{\prime}(k|\mathbf{y}|)-J_{n}(k|\mathbf{y}|)\right) \nabla_{\mathbb{S}} e^{-\mathrm{i} n \theta_{\mathbf{y}}}\right).
\end{aligned}
\end{equation}
From the identities \eqref{eq:de2} and \eqref{eq:der}, and together with the help of Lemma \ref{lem:in}, one can obtain that
\begin{equation*}
\begin{split}
& \int_{\partial B_{R}} \nabla_{\mathbf{x}} \nabla_{\mathbf{x}} H_{0}(k|\mathbf{x}-\mathbf{y}|) e^{\rmi n \theta_{\mathbf{y}}} \bnu_{\mathbf{y}} d s_{\by}= \\
& \frac{\pi H_{n-1}(k|\mathbf{x}|) e^{\mathrm{i}(n-1) \theta}}{R}\left(k^{2} R^{2} J_{n-1}^{\prime \prime}(k R)-\left(k R J_{n-1}^{\prime}(k R)-J_{n-1}(k R)\right)(n-1)\right)\left(\begin{array}{c}
1 \\
\mathrm{i}
\end{array}\right)+ \\
& \frac{\pi H_{n+1}(k|\mathbf{x}|) e^{\mathrm{i}(n+1) \theta}}{R}\left(k^{2} R^{2} J_{n+1}^{\prime \prime}(k R)+\left(k R J_{n+1}^{\prime}(k R)-J_{n+1}(k R)\right)(n+1)\right)\left(\begin{array}{c}
1 \\
-\mathrm{i}
\end{array}\right).
\end{split}
\end{equation*}
Now using the identity \eqref{eq:bd}, we can simplify the last equation as 
\begin{equation*}
\begin{split}
& \int_{\partial B_{R}} \nabla_{\mathbf{x}} \nabla_{\mathbf{x}} H_{0}(k|\mathbf{x}-\mathbf{y}|) e^{\mathrm{i} n \theta_{\mathbf{y}}} \bnu_{\mathbf{y}} d s_{\by}= \\
& \frac{\pi H_{n-1}(k|\mathbf{x}|) e^{\mathrm{i}(n-1) \theta}}{R}\left(\left(n(n-1)-k^{2} R^{2}\right) J_{n-1}(k R)-n k R J_{n-1}^{\prime}(k R)\right)\left(\begin{array}{c}
1 \\
\mathrm{i}
\end{array}\right)+ \\
&\frac{\pi H_{n+1}(k|\mathbf{x}|) e^{\mathrm{i}(n+1) \theta}}{R}\left(\left(n(n+1)-k^{2} R^{2}\right) J_{n+1}(k R)+n k R J_{n+1}^{\prime}(k R)\right)\left(\begin{array}{c}
1 \\
-\mathrm{i}
\end{array}\right).
\end{split}
\end{equation*}
The proof is completed.
\end{proof}

\begin{lem}\label{lem:gt}
There holds the following identity for $\bx\in\mathbb{R}^2\backslash\overline{B}_R$:
\begin{equation}
\begin{split}
& \int_{\partial B_{R}} \nabla_{\mathbf{x}} \nabla_{\mathbf{x}} H_{0}(k|\mathbf{x}-\mathbf{y}|) e^{\mathrm{i} n \theta_{\mathbf{y}}} \mathbf{t}_{\mathbf{y}} d s_{\by}= \\
& \frac{-\mathrm{i} n \pi H_{n-1}(k|\mathbf{x}|) e^{\mathrm{i}(n-1) \theta}}{R}\left(k R J_{n-1}^{\prime}(k R)-J_{n-1}(k R)(n-1)\right)\left(\begin{array}{c}
1 \\
\mathrm{i}
\end{array}\right)+ \\
& \frac{-\mathrm{i} n \pi H_{n+1}(k|\mathbf{x}|) e^{\mathrm{i}(n+1) \theta}}{R}\left(k R J_{n+1}^{\prime}(k R)+J_{n+1}(k R)(n+1)\right)\left(\begin{array}{c}
1 \\
-\mathrm{i}
\end{array}\right).
\end{split}
\end{equation}

\end{lem}

\begin{proof}
First note that $\nabla_{\mathbf{x}} \nabla_{\mathbf{x}} H_{0}(k|\mathbf{x}-\mathbf{y}|) = \nabla_{\mathbf{y}} \nabla_{\mathbf{y}} H_{0}(k|\mathbf{x}-\mathbf{y}|)$. Then one has that

\begin{equation}\label{eq:de2t}
\begin{aligned}
& \int_{\partial B_{R}} \nabla_{\mathbf{x}} \nabla_{\mathbf{x}} H_{0}(k|\mathbf{x}-\mathbf{y}|) e^{\mathrm{i} n \theta_{\mathbf{y}}} \mathbf{t}_{\mathbf{y}} d s_{\by}=\int_{\partial B_{R}} \nabla_{\mathbf{y}} \nabla_{\mathbf{y}} H_{0}(k|\mathbf{x}-\mathbf{y}|) e^{\rmi n \theta_{\mathbf{y}}} \mathbf{t}_{\mathbf{y}} d s_{\by} \\
=& \int_{\partial B_{R}}\left(\nabla_{\mathbb{S}} \nabla_{\mathbf{y}} H_{0}(k|\mathbf{x}-\mathbf{y}|) /|\mathbf{y}|+\partial_{r} \nabla_{\mathbf{y}} H_{0}(k|\mathbf{x}-\mathbf{y}|) \nu_{\mathbf{y}}\right) \frac{1}{\rmi n} \nabla_{\mathbb{S}} e^{\rmi n \theta_{\mathbf{y}}} d s_{\by} \\
=&-\frac{1}{\mathrm{i} n} \int_{\mathbb{S}}\left(\nabla_{\mathbf{y}} H_{0}(k|\mathbf{x}-\mathbf{y}|)\right) \Delta_{S} e^{\rmi n \theta_{\mathbf{y}}} d s_{\by} \\
=&-\rmi n \int_{\mathbb{S}} \nabla_{\mathbf{y}} H_{0}(k|\mathbf{x}-\mathbf{y}|) e^{\mathrm{i} n \theta_{\mathbf{y}}} d s_{\by}.
\end{aligned}
\end{equation}
Then, from the expansion of the function $H_{0}(k|\mathbf{x}-\mathbf{y}|)$ in \eqref{eq:expH}, it holds that
\begin{equation}\label{eq:dert}
\begin{aligned}
& \nabla_{\mathbf{y}} H_{0}(k|\mathbf{x}-\mathbf{y}|)=\nabla_{\mathbf{y}}\left(\sum_{n \in \mathbb{Z}} H_{n}(k|\mathbf{x}|) e^{\rmi n \theta_{\mathbf{x}}} J_{n}(k|\mathbf{y}|) e^{-\mathrm{i} n \theta_{\mathbf{y}}}\right) \\
& =\sum_{n \in \mathbb{Z}} \frac{H_{n}(k|\mathbf{x}|) e^{\mathrm{i} n \theta_{\mathbf{x}}}}{|\mathbf{y}|} \left(k|\mathbf{y}| J_{n}^{\prime}(k|\mathbf{y}|) e^{-\mathrm{i} n \theta_{\mathbf{y}}} \bnu_{\mathbf{y}}+J_{n}(k|\mathbf{y}|) \nabla_{\mathbb{S}} e^{-\mathrm{i} n \theta_{\mathbf{y}}}\right).
\end{aligned}
\end{equation}
From the equations \eqref{eq:de2t} and \eqref{eq:dert}, and together with the help of Lemma \ref{lem:in}, one can obtain that
\begin{equation*}
\begin{split}
& \int_{\partial B_{R}} \nabla_{\mathbf{x}} \nabla_{\mathbf{x}} H_{0}(k|\mathbf{x}-\mathbf{y}|) e^{\rmi n \theta_{\mathbf{y}}} \mathbf{t}_{\mathbf{y}} d s_{\by}\\
= & \frac{-\mathrm{i} n \pi H_{n-1}(k|\mathbf{x}|) e^{\mathrm{i}(n-1) \theta}}{R}\left(k R J_{n-1}^{\prime}(k R)-J_{n-1}(k R)(n-1)\right)\left(\begin{array}{c}
1 \\
\mathrm{i}
\end{array}\right) \\
& + \frac{-\mathrm{i} n \pi H_{n+1}(k|\mathbf{x}|) e^{\mathrm{i}(n+1) \theta}}{R}\left(k R J_{n+1}^{\prime}(k R)+J_{n+1}(k R)(n+1)\right)\left(\begin{array}{c}
1 \\
-\mathrm{i}
\end{array}\right).
\end{split}
\end{equation*}

\end{proof}

With these preparations, we can present the expressions of the single-layer potentials $\mathbf{S}_{\partial B_{R}}^{\omega}$ with two densities $e^{\mathrm{i} n \theta} \boldsymbol{\nu}$ and $e^{\mathrm{i} n \theta} \mathbf{t}$, 
and the proof follows directly from the definition of the single-layer potential operator $\mathbf{S}_{\partial B_{R}}^{\omega}$ in \eqref{eq:single} and Lemmas \ref{lem:d1}, \ref{lem:gnu} and \ref{lem:gt}.

\begin{thm}\label{thm:sing}
The single layer potentials  $\mathbf{S}_{\partial B_{R}}^{\omega}\left[e^{\mathrm{i} n \theta} \boldsymbol{\nu}\right]$ and $\mathbf{S}_{\partial B_{R}}^{\omega}\left[e^{\mathrm{i} n \theta} \mathbf{t}\right]$have the following expressions for $\bx\in\mathbb{R}^2\backslash\overline{B}_R$, 
\begin{equation}
\begin{split}
& \mathbf{S}_{\partial B_{R}}^{\omega} [e^{\mathrm{i} n \theta} \boldsymbol{\nu} ](\bx)= \\
& \frac{-\mathrm{i} \pi e^{\mathrm{i}(n-1) \theta}}{4 \omega^{2} R}\left(n H_{n-1}\left(k_{s}|\mathbf{x}|\right)\left((n-1) J_{n-1}\left(k_{s} R\right)-k_{s} R J_{n-1}^{\prime}\left(k_{s} R\right)\right)+\right. \\
& \left.H_{n-1}\left(k_{p}|\mathbf{x}|\right)\left(\left(k_{p}^{2} R^{2}-n^{2}+n\right) J_{n-1}\left(k_{p} R\right)+n k_{p} R J_{n-1}^{\prime}\left(k_{p} R\right)\right)\right)\left(\begin{array}{c}
1 \\
\mathrm{i}
\end{array}\right)+ \\
& \frac{-\mathrm{i} \pi e^{\mathrm{i}(n+1) \theta}}{4 \omega^{2} R}\left(n H_{n+1}\left(k_{s}|\mathbf{x}|\right)\left((n+1) J_{n+1}\left(k_{s} R\right)+k_{s} R J_{n+1}^{\prime}\left(k_{s} R\right)\right)+\right. \\
& \left.H_{n+1}\left(k_{p}|\mathbf{x}|\right)\left(\left(k_{p}^{2} R^{2}-n^{2}-n\right) J_{n+1}\left(k_{p} R\right)-n k_{p} R J_{n+1}^{\prime}\left(k_{p} R\right)\right)\right)\left(\begin{array}{c}
1 \\
-\mathrm{i}
\end{array}\right),
\end{split}
\end{equation}
and 
\begin{equation}
\begin{split}
&\mathbf{S}_{\partial B_{R}}^{\omega} [e^{\mathrm{i} n \theta} \mathbf{t} ](\bx)= \\
& \frac{-\pi e^{\mathrm{i}(n-1) \theta}}{4 \omega^{2} R}\left(H_{n-1}\left(k_{s}|\mathbf{x}|\right)\left(\left(k_{s}^{2} R^{2}-n^{2}+n\right) J_{n-1}\left(k_{s} R\right)+n k_{s} R J_{n-1}^{\prime}\left(k_{s} R\right)\right)+\right. \\
& \left.n H_{n-1}\left(k_{p}|\mathbf{x}|\right)\left((n-1) J_{n-1}\left(k_{p} R\right)-k_{p} R J_{n-1}^{\prime}\left(k_{p} R\right)\right)\right)\left(\begin{array}{c}
1 \\
\mathrm{i}
\end{array}\right)+ \\
& \frac{\pi e^{\mathrm{i}(n+1) \theta}}{4 \omega^{2} R}\left(H_{n+1}\left(k_{s}|\mathbf{x}|\right)\left(\left(k_{s}^{2} R^{2}-n^{2}-n\right) J_{n+1}\left(k_{s} R\right)-n k_{s} R J_{n+1}^{\prime}\left(k_{s} R\right)\right)+\right. \\
& \left.n H_{n+1}\left(k_{p}|\mathbf{x}|\right)\left((n+1) J_{n+1}\left(k_{p} R\right)+k_{p} R J_{n+1}^{\prime}\left(k_{p} R\right)\right)\right)\left(\begin{array}{c}
1 \\
-\mathrm{i}
\end{array}\right).
\end{split}
\end{equation}
\end{thm}
%

\begin{rem}
With the help of the recursion formulas in \eqref{eq:re}, the single-layer potentials  $\mathbf{S}_{\partial B_{R}}^{\omega}\left[e^{\mathrm{i} n \theta} \boldsymbol{\nu}\right]$ and $\mathbf{S}_{\partial B_{R}}^{\omega}\left[e^{\mathrm{i} n \theta} \mathbf{t}\right]$ can be expressed as follows for $\bx\in\mathbb{R}^2\backslash\overline{B}_R$:
\begin{align}
 \mathbf{S}_{\partial B_{R}}^{\omega} [e^{\mathrm{i} n \theta} \boldsymbol{\nu} ](\bx)=&  \frac{-\mathrm{i} \pi }{4 \omega^{2} R}\left(n  k_{s} R J_{n} \left(k_{s} R\right) \mathbf{Q}_n^o(k_s |\bx|)  +  k_{p}^2 R^2   J_{n}^{\prime}\left(k_{p} R \right) \mathbf{P}_n^o(k_p |\bx|) \right),\\ 
\mathbf{S}_{\partial B_{R}}^{\omega} [e^{\mathrm{i} n \theta} \mathbf{t} ](\bx)=&  \frac{-\pi }{4 \omega^{2} R}\left( k_s^2 R^2   J_{n}^{\prime}\left(k_{s} R \right) \mathbf{Q}_n^o(k_s |\bx|)   + n k_p R J_n(k_p R) \mathbf{P}_n^o(k_p |\bx|)  \right),
\end{align}
where
\begin{align}
\mathbf{Q}_n^o(k_s |\bx|) = & \frac{2 n H_n(k_s |\bx|)}{k_s |\bx|}  e^{\mathrm{i} n \theta} \bnu + 2\rmi  H_{n}^{\prime}(k_s |\bx|)  e^{\mathrm{i} m \theta} \mathbf{t},\label{eq:q}\\
\mathbf{P}_n^o(k_p |\bx|) =& 2 H_{n}^{\prime}(k_p |\bx|)  e^{\mathrm{i} n \theta} \bnu +  \frac{2 \rmi n H_n(k_p |\bx|)}{k_p |\bx|}  e^{\mathrm{i} m \theta} \mathbf{t}.\label{eq:p}
\end{align}
Moreover, these two functions $\mathbf{Q}_n^o(k_s|\bx|)$ and $\mathbf{P}_n^o(k_p|\bx|)$ are radiating solutions to the equation $\left( \Lcal_{\lambda,\mu}+\omega^2\right)\bu=0$ in $\bx\in\mathbb{R}^2\backslash\overline{B}_R$. The function $\mathbf{Q}_n^o(k_s|\bx|)$ belongs to the s-wave and the function $\mathbf{P}_n^o(k_p|\bx|)$ belongs to the p-wave.
\end{rem}

Following similar deductions to the above, one can derive the following proposition.
\begin{prop}\label{pro:sinin}
For $\bx\in B_R$, the single layer potentials  $\mathbf{S}_{\partial B_{R}}^{\omega}\left[e^{\mathrm{i} n \theta} \boldsymbol{\nu}\right]$ and $\mathbf{S}_{\partial B_{R}}^{\omega}\left[e^{\mathrm{i} n \theta} \mathbf{t}\right]$ have the following expressions:
\begin{equation}
\begin{split}
 \mathbf{S}_{\partial B_{R}}^{\omega} [e^{\mathrm{i} n \theta} \boldsymbol{\nu} ](\bx)= 
 \frac{-\mathrm{i} \pi  }{4 \omega^{2} R}\left( n  k_{s} R  H_{n} \left(k_{s} R\right) \mathbf{Q}_n^i (k_s |\bx|)  +  k_{p}^2 R^2 H_{n}^{\prime}\left(k_{p} R \right) \mathbf{P}_n^i(k_p |\bx|) \right),\\
\mathbf{S}_{\partial B_{R}}^{\omega} [e^{\mathrm{i} n \theta} \mathbf{t} ](\bx)=  
\frac{-\pi }{4 \omega^{2} R}\left( k_s^2 R^2  H_{n}^{\prime}\left(k_{s} R \right) \mathbf{Q}_n^i(k_s |\bx|)   + n k_p R H_n(k_p R) \mathbf{P}_n^i (k_p |\bx|)  \right),
\end{split}
\end{equation}
where
\begin{align}
\mathbf{Q}_n^i (k_s |\bx|) =& \frac{2 n J_n(k_s |\bx|)}{k_s |\bx|}  e^{\mathrm{i} n \theta} \bnu + 2\rmi  J_{n}^{\prime}(k_s |\bx|)  e^{\mathrm{i} n \theta} \mathbf{t},\label{eq:qi}\\
\mathbf{P}_n^i (k_p |\bx|) =& 2 J_{n}^{\prime}(k_p |\bx|)  e^{\mathrm{i} n \theta} \bnu +  \frac{2 \rmi n J_n(k_p |\bx|)}{k_p |\bx|}  e^{\mathrm{i} n \theta} \mathbf{t}.\label{eq:pi}
\end{align}
Moreover, these two functions $\mathbf{Q}_n^i(k_s|\bx|)$ and $\mathbf{P}_n^i(k_p|\bx|)$ are entire solutions to the equation $\left( \Lcal_{\lambda,\mu}+\omega^2\right)\bu=0$ in $\bx\in B_R$. The function $\mathbf{Q}_n^i(k_s|\bx|)$ belongs to the s-wave and the function $\mathbf{P}_n^i(k_p|\bx|)$ belongs to the p-wave.
\end{prop}

Since the single layer potential operator $\mathbf{S}_{\partial B_{R}}^{\omega}$ is continuous from $\mathbb{R}^2\backslash\overline{B}_R$ to $\mathbb{R}^2\backslash{B}_R$, by letting $|\bx|=R$ in Theorem \ref{thm:sing} and together with the help of recursion formulas given in \eqref{eq:re}, one has the following lemma.

\begin{lem}\label{lem:eisin}
The single layer potentials  $\mathbf{S}_{\partial B_{R}}^{\omega}\left[e^{\mathrm{i} n \theta} \boldsymbol{\nu}\right]$ and $\mathbf{S}_{\partial B_{R}}^{\omega}\left[e^{\mathrm{i} n \theta} \mathbf{t}\right]$ have the following expressions for $|\bx|=R$:
\begin{equation}
 \mathbf{S}_{\partial B_{R}}^{\omega} [e^{\mathrm{i} n \theta} \bnu ](\bx) = \alpha_{1 n} e^{\mathrm{i} n \theta} \bnu+\alpha_{2 n} e^{\mathrm{i} m \theta} \mathbf{t}\quad\mbox{and}\quad\mathbf{S}_{\partial B_{R}}^{\omega} [e^{\mathrm{i} n \theta} \mathbf{t} ](\bx)= \alpha_{3 n} e^{\mathrm{i} n \theta} \bnu+\alpha_{4 n} e^{\mathrm{i} m \theta} \mathbf{t},
\end{equation}
where
\begin{eqnarray*}
   \alpha_{1 n} & = & -\frac{\rmi \pi}{2\omega^2 R} \left( n^2 J_n(k_s R) H_n(k_s R) + k_p^2 R^2 J^{\prime}_n(k_p R) H^{\prime}_n(k_p R)  \right),\\
  \alpha_{2 n} & = & \frac{n \pi}{2\omega^2 } \left(k_s J_n(k_s R) H^{\prime}_n(k_s R) + k_p J^{\prime}_n(k_p R) H_n(k_p R)  \right), \\
  \alpha_{3 n} & = & -\frac{n \pi}{2\omega^2 } \left(k_s J^{\prime}_n(k_s R) H_n(k_s R) + k_p J_n(k_p R) H^{\prime}_n(k_p R)  \right), \\
  \alpha_{4 n} & = & -\frac{\rmi \pi}{2\omega^2 R} \left( k_s^2 R^2 J^{\prime}_n(k_s R) H^{\prime}_n(k_s R) + n^2 J_n(k_p R) H_n(k_p R)  \right).
\end{eqnarray*}

\end{lem}

%

Next, we calculate the tractions $\partial_{\bnu} \mathrm{S}_{\partial B_{R}}^{\omega}\left[e^{\mathrm{i} n \theta} \boldsymbol{\bnu}\right]|_{\pm}$ and $\partial_{\bnu} \mathrm{S}_{\partial B_{R}}^{\omega}\left[e^{\mathrm{i} n \theta} \mathbf{t}\right]|_{\pm}$ on the boundary $\partial B_R$, where the traction
operator $\partial_{\bnu}$ is defined in \eqref{eq:trac}. First, we notice that
\begin{equation}\label{eq:dp12}
\frac{\partial}{\partial_{x_{1}}}\left(H_{n}(k|\bx|) e^{\mathrm{i} n \theta}\right)=p_{1 n} \quad \text { and } \quad \frac{\partial}{\partial_{x_{2}}}\left(H_{n}(k|\bx|) e^{\mathrm{i} n \theta}\right)=p_{2 n},
\end{equation}
where
\begin{eqnarray*}
p_{1 n} &=& k H_{n}^{\prime}(k|\bx|) e^{\mathrm{i} n \theta} \cos (\theta)-\mathrm{i} n H_{n}(k|\bx|) e^{\mathrm{i} n \theta} \sin (\theta) /|\mathbf{x}|,\\
p_{2 n} &=& k H_{n}^{\prime}(k|\bx|) e^{\mathrm{i} n \theta} \sin (\theta)+\mathrm{i} n H_{n}(k|\bx|) e^{\mathrm{i} n \theta} \cos (\theta) /|\mathbf{x}|.
\end{eqnarray*}
Hence for
\[
\mathbf{g}=e^{\mathrm{i} n \theta} H_{n}(k|\mathbf{x}|)\left(\begin{array}{c}
a \\
b
\end{array}\right),
\]
where $a$ and $b$ are two constants, one has that
\[
\nabla \cdot \mathbf{g}=a p_{1 n}+b p_{2 n},\qquad 2 \nabla^{s} \mathbf{g}=\left(\begin{array}{cc}
2 a p_{1 n} & a p_{2 n}+b p_{1 n} \\
a p_{2 n}+b p_{1 n} & 2 b p_{2 n}
\end{array}\right),
\]
where $p_{1 n}$ and $p_{2 n}$ are defined in \eqref{eq:dp12}. During the simplification, we have used the recursion formulas given in \eqref{eq:re}. Then we can obtain the follow lemma.

\begin{lem}\label{lem:pnu}
There hold the following relations:
\begin{equation}
\begin{split}
 \partial_{\bnu} \mathbf{S}_{\partial B_{R}}^{\omega}\left[e^{\mathrm{i} m \theta} \boldsymbol{\nu}\right]|_+ & =g_{1, m}(|\mathbf{x}|) e^{\mathrm{i} m \theta} \boldsymbol{\nu}+g_{2, m}(|\mathbf{x}|) e^{\mathrm{i} m \theta} \mathbf{t}, \\
 \partial_{\bnu} \mathbf{S}_{\partial B_{R}}^{\omega}\left[e^{\mathrm{i} m \theta} \mathbf{t}\right]|_+ & =g_{3, m}(|\mathbf{x}|) e^{i m \theta} \boldsymbol{\nu}+g_{4, m}(|\mathbf{x}|) e^{i m \theta} \mathbf{t},
\end{split}
\end{equation}
where

\begin{align*}
  g_{1, m}(|\mathbf{x}|)= &\frac{\mathrm{i} \pi}{2 \omega^{2} R^{2}}\left(2 \mu m^{2} J_{m}\left(k_{s} R\right)\left(H_{m}\left(k_{s}|\mathbf{x}|\right)-k_{s} R H_{m}^{\prime}\left(k_{s}|\mathbf{x}|\right)\right)+\right. \\
&\left.J_{m}^{\prime}\left(k_{p} R\right) k_{p} R\left(H_{m}\left(k_{p}|\mathbf{x}|\right)\left(\omega^{2} R^{2}-2 \mu m^{2}\right)+2 k_{p} \mu R H_{m}^{\prime}\left(k_{p}|\mathbf{x}|\right)\right)\right), \\
g_{2, m}(|\mathbf{x}|) =&-\frac{m \mu \pi}{2 \omega^{2} R^2}\left(J_{m}\left(k_{s} R\right) H_{m}\left(k_{s}|\mathbf{x}|\right)\left(k_{s}^{2} R^{2}-2 m^{2}\right)+\right. \\
 & \left.2 R\left(k_{s} J_{m}\left(k_{s} R\right) H_{m}^{\prime}\left(k_{s}|\mathbf{x}|\right)+k_{p} J_{m}^{\prime}\left(k_{p} R\right)\left(H_{m}\left(k_{p}|\mathbf{x}|\right)-k_{p} R H_{m}^{\prime}\left(k_{p}|\mathbf{x}|\right)\right)\right)\right),
\end{align*}

\begin{align*}
g_{3, m}(|\mathbf{x}|)= & \frac{m \pi}{2 \omega^{2} R^2}\left(J_{m}\left(k_{p} R\right) H_{m}\left(k_{p}|\mathbf{x}|\right)\left((\lambda+2 \mu) k_{p}^{2} R^{2}-2 \mu m^{2}\right)+\right. \\
 & \left.2 \mu R\left(k_{p} J_{m}\left(k_{p} R\right) H_{m}^{\prime}\left(k_{p}|\mathbf{x}|\right)+k_{s} J_{m}^{\prime}\left(k_{s} R\right)\left(H_{m}\left(k_{s}|\mathbf{x}|\right)-k_{s} R H_{m}^{\prime}\left(k_{s}|\mathbf{x}|\right)\right)\right)\right), \\
g_{4, m}(|\mathbf{x}|)= &\frac{\mathrm{i} \mu \pi}{2 \omega^{2} R^{2}}\left(2 m^{2} J_{m}\left(k_{p} R\right)\left(H_{m}\left(k_{p}|\mathbf{x}|\right)-k_{p} R H_{m}^{\prime}\left(k_{p}|\mathbf{x}|\right)\right)+\right. \\
 &\left.J_{m}^{\prime}\left(k_{s} R\right) k_{s} R\left(k_{s}^{2} R^{2} H_{m}\left(k_{s}|\mathbf{x}|\right)+2 k_{s} R H_{m}^{\prime}\left(k_{s}|\mathbf{x}|\right)-2 m^{2} H_{m}\left(k_{s}|\mathbf{x}|\right)\right)\right).
\end{align*}

\end{lem}

%
\begin{rem}
Taking the traction of the function $\mathbf{Q}_n^i$ and $\mathbf{P}_n^i$ defined in \eqref{eq:qi} and \eqref{eq:pi} on the boundary $\partial B_R$ gives that 
\begin{equation}\label{eq:tpi}
\begin{split}
 \partial_{\bnu} \mathbf{Q}_n^i =  \gamma_{1n} e^{\rmi n \theta}\bnu + \gamma_{2n} e^{\rmi n \theta}\bt, \ \ \partial_{\bnu} \mathbf{P}_n^i = \gamma_{3n} e^{\rmi n \theta}\bnu +  \gamma_{4n} e^{\rmi n \theta}\bt,
\end{split}
\end{equation}
where
\[
 \gamma_{1n} = \frac{4n\mu}{k_s R^2} \left( k_{s} R J_{n}^{\prime}\left(k_{s}R\right) -  J_{n}\left(k_{s}R\right)  \right), \ \  \gamma_{2n} =\frac{2 \rmi \mu}{k_s R^2} \left( \left(2n^2 - k_{s}^{2} R^{2}\right) J_{n}\left(k_{s}R\right)-2 k_{s} R J_{n}^{\prime}\left(k_{s}R\right) \right),
\]
\[
 \gamma_{3n}= \frac{2 \mu}{k_p R^2} \left( \left(2n^2 - k_{s}^{2} R^{2}\right) J_{n}\left(k_{p}R\right)  -2  k_p R J_{n}^{\prime}(k_p R) \right) , \ \ \gamma_{4n} = \frac{4\rmi n\mu}{k_p R^2}\left( k_p R J_{n}^{\prime}(k_p R) - J_n(k_p R)  \right).
\]
\end{rem}

With the help of Lemma \ref{lem:pnu} and the jump formula in \eqref{eq:jump}, one can conclude the following lemma.

\begin{lem}\label{lem:npr}
There hold that
\[
\mathbf{K}_{\partial B_{R}}^{\omega, *}\left[e^{\mathrm{i} n \theta} \boldsymbol{\nu}\right]=a_{1 n} e^{\mathrm{i} n \theta} \boldsymbol{\nu}+a_{2 n} e^{\mathrm{i} n \theta} \mathbf{t}\quad\mbox{and}\quad
\mathbf{K}_{\partial B_{R}}^{\omega, *}\left[e^{\mathrm{i} n \theta} \mathbf{t}\right]=b_{1 n} e^{\mathrm{i} n \theta} \boldsymbol{\nu}+b_{2 n} e^{\mathrm{i} n \theta} \mathbf{t},
\]
where
\[
a_{1 n}=-\frac{1}{2}+g_{1, n}(R), \quad a_{2 n}=g_{2, n}(R), \quad b_{1 n}=g_{3, n}(R), \quad b_{2 n}=-\frac{1}{2}+g_{4, n}(R);
\]
with the functions $g_{i, n}(|\mathbf{x}|), 1 \leq i \leq 4$ given in Lemma \ref{lem:pnu}.
\end{lem}

Finally, we obtain the eigensystem for the N-P operator $\mathbf{K}_{\partial B_{R}}^{\omega, *}$.
\begin{thm}\label{thm:ei}
Let $a_{1 n}, a_{2 n}, b_{1 n}, b_{2 n}$ be given in Lemma \ref{lem:npr}.
The eigensystem for the N-P operator $\mathbf{K}_{\partial B_{R}}^{\omega, *}$ is given as follows: \\
1) if $a_{2 n} \neq 0$, the eigenvalues are 
\[
\begin{split}
\xi_{1 n} & =\frac{1}{2}\left(a_{1 n}+b_{2 n}-\sqrt{a_{1 n}^{2}-2 a_{1 n} b_{2 n}+4 a_{2 n} b_{1 n}+b_{2 n}^{2}}\right), \\
\xi_{2 n} & =\frac{1}{2}\left(a_{1 n}+b_{2 n}+\sqrt{a_{1 n}^{2}-2 a_{1 n} b_{2 n}+4 a_{2 n} b_{1 n}+b_{2 n}^{2}}\right),
\end{split}
\]
and the corresponding eigenfunctions are 
\[
\begin{aligned}
\mathbf{p}_{1 n} &=\left(\frac{a_{1 n}-b_{2 n}-\sqrt{a_{1 n}^{2}-2 a_{1 n} b_{2 n}+4 a_{2 n} b_{1 n}+b_{2 n}^{2}}}{2 a_{2 n}} \right) e^{\mathrm{i} n \theta} \boldsymbol{\nu} + e^{\mathrm{i} n \theta} \boldsymbol{t} , \\
\mathbf{p}_{2 n} &=\left(\frac{a_{1 n}-b_{2 n}+\sqrt{a_{1 n}^{2}-2 a_{1 n} b_{2 n}+4 a_{2 n} b_{1 n}+b_{2 n}^{2}}}{2 a_{2 n}} \right) e^{\mathrm{i} n \theta} \boldsymbol{\nu} + e^{\mathrm{i} n \theta} \boldsymbol{t}  ;
\end{aligned}
\]
2) if $a_{2 n}=0,$ and $a_{1 n} \neq b_{2 n}$,  the eigenvalues are 
\[
\xi_{1 n}=a_{1 n}, \quad \xi_{2 n}=b_{2 n},
\]
and the corresponding eigenfunctions are 
\[
\mathbf{p}_{1 n}=e^{\mathrm{i} n \theta}\bnu, \quad \mathbf{p}_{2 n}=\left(\frac{b_{1 n}}{b_{2 n}-a_{1 n}} \right) e^{\mathrm{i} n \theta}\bnu + e^{\mathrm{i} n \theta} \bt;
\]
3) if $a_{2 n}=0, a_{1 n}=b_{2 n}$ and $b_{1 n}=0$,  the eigenvalues are 
\[
\xi_{1 n}=a_{1 n}, \quad \xi_{2 n}=a_{1 n},
\]
and the corresponding eigenfunctions are 
\[
\mathbf{p}_{1 n}=e^{\mathrm{i} n \theta}\bnu, \quad \mathbf{p}_{2 n}= e^{\mathrm{i} n \theta} \bt;
\]
4) if $a_{2 n}=0, a_{1 n}=b_{2 n}$ and $b_{1 n} \neq 0$, the eigenvalues are 
\[
\xi_{1 n}=a_{1 n}, \quad \xi_{2 n}=a_{1 n},
\]
there is only one eigenfunction 
\[
\mathbf{p}_{1 n}=e^{\mathrm{i} n \theta}\bnu,
\]
and another one is the generalized eigenfunction,
\[
\mathbf{p}_{2 n}= \frac{1}{b_{1n}}  e^{\mathrm{i} n \theta} \bt,
\]
namely, $\mathbf{p}_{2 n}$ satisfies 
\[
 \left( \mathbf{K}_{\partial B_{R}}^{\omega, *} - \xi_{1n} \right)\mathbf{p}_{2 n} = \mathbf{p}_{1 n}.
\]
\end{thm}

\begin{proof}
We first know from Lemma \ref{lem:npr} that
\begin{equation}\label{eq:trama}
\mathbf{K}_{\partial B_{R}}^{\omega, *}[\boldsymbol{\nu}, \mathbf{t}]=(\boldsymbol{\nu}, \mathbf{t}) T_{n},
\end{equation}
where $T_{n}$ is a $2 \times 2$ matrix given by $T_{n}=\big(a_{1 n},  b_{1 n};  a_{2 n}, b_{2 n}\big)$
Thus, we focus ourself on investigating the eigensystem of the matrix $T_{n} $, which could further lead to the eigensystem of the operator $\mathbf{K}_{\partial B_{R}}^{\omega, *}$.
Specifically, we like to find the matrix $P_{n}=\left(\mathbf{p}_{1 n}, \mathbf{p}_{2 n}\right)$ 
such that 
\begin{equation}\label{eq:mei}
T_{n} P_{n}=P_{n} \Lambda_{n},
\end{equation}
where the matrix $\Lambda_{n}$ is an diagonal matrix, namely $\Lambda_{n}=\big(\xi_{1 n}, 0; 0, \xi_{2 n}\big)$.
A direct calculation shows that if $a_{2 n} \neq 0,$ one has that
\[
\begin{aligned}
\mathbf{p}_{1 n} &=\left(\frac{a_{1 n}-b_{2 n}-\sqrt{a_{1 n}^{2}-2 a_{1 n} b_{2 n}+4 a_{2 n} b_{1 n}+b_{2 n}^{2}}}{2 a_{2 n}}, 1\right)^{t}, \\
\mathbf{p}_{2 n} &=\left(\frac{a_{1 n}-b_{2 n}+\sqrt{a_{1 n}^{2}-2 a_{1 n} b_{2 n}+4 a_{2 n} b_{1 n}+b_{2 n}^{2}}}{2 a_{2 n}}, 1\right)^{t},
\end{aligned}
\]
\[
\xi_{1 n}=\frac{1}{2}\left(a_{1 n}+b_{2 n}-\sqrt{a_{1 n}^{2}-2 a_{1 n} b_{2 n}+4 a_{2 n} b_{1 n}+b_{2 n}^{2}}\right),
\]
\[
\xi_{2 n}=\frac{1}{2}\left(a_{1 n}+b_{2 n}+\sqrt{a_{1 n}^{2}-2 a_{1 n} b_{2 n}+4 a_{2 n} b_{1 n}+b_{2 n}^{2}}\right).
\]
For the case $a_{2 n}=0$ and $a_{1 n} \neq b_{2 n},$ one has that
\[
\mathbf{p}_{1 n}=(1,0)^{t}, \quad \mathbf{p}_{2 n}=\left(\frac{b_{1 n}}{b_{2 n}-a_{1 n}}, 1\right)^{t}\quad\mbox{and}\quad
\xi_{1 n}=a_{1 n}, \quad \xi_{2 n}=b_{2 n}.
\]
Moreover, if $a_{2 n}=0, a_{1 n}=b_{2 n}$ and $b_{1 n}=0,$ one has that
\[
\mathbf{p}_{1 n}=(0,1)^{t}, \quad \mathbf{p}_{2 n}=(1,0)^{t}, \quad \xi_{1 n}=a_{1 n}, \quad \xi_{2 n}=a_{1 n}.
\]
For the last case $a_{2 n}=0, a_{1 n}=b_{2 n}$ and $b_{1 n} \neq 0$, the situation is different. The matrix $\Lambda_{n}$ given in \eqref{eq:mei} is not a diagonal matrix anymore, but a Jordan matrix, given as follows
\[
\Lambda_{n}=\left(\begin{array}{cc}
a_{1 n} & 1 \\
0 & a_{2 n}
\end{array}\right).
\]
Then the generalized eigenvectors are given as $\mathbf{p}_{1 n}=(1,0)^{t}$ and $\mathbf{p}_{2 n}=\left(0,  1/b_{1n} \right)^{t}$.
Finally, with the help of the relationship \eqref{eq:trama}, one can prove the statement of the theorem and the proof is completed.
\end{proof}

\begin{rem}
We present the asymptotic expansion for the eigenvalues when the frequency $\omega \ll 1 .$ From the asymptotic expansions of the Bessel function and Hankel function in \eqref{eq:asj} and \eqref{eq:ash1} for $\omega \ll 1,$ one has that when $|n| \geq 2$,
\[
 |a_{2n}|=\frac{\mu}{2(\lambda+2\mu)} + \mathcal{O}(\omega)\neq 0,
\]
which is the first case in Theorem \ref{thm:ei}, thus the eigenvalues are
\[
\xi_{1 n}=-\frac{\mu}{2(\lambda+2 \mu)}+\mathcal{O}\left(\omega^{2}\right), \quad \xi_{2 n}=\frac{\mu}{2(\lambda+2 \mu)}+\mathcal{O}\left(\omega^{2}\right).
\]
When $|n|=1,$ one has that
\[
 |a_{2n}|=\frac{\mu}{2(\lambda+2\mu)} + \mathcal{O}(\omega)\neq 0,
\]
which is the first case in Theorem \ref{thm:ei}, thus the eigenvalues are
\[
\xi_{1 n}=\frac{\mu}{2(\lambda+2 \mu)}+o(\omega), \quad \xi_{2 n}=\frac{1}{2}+\mathcal{O}\left(\omega^{2}\right),
\]
When $n=0,$ one has that
\[
a_{2 n}=b_{1 n}=0,  \quad  \mbox{and}\quad a_{1 n} \neq b_{2 n},
\]
which is the second case in Theorem \ref{thm:ei}, thus the eigenvalues are
\[
\xi_{1 n}=-\frac{\lambda}{2(\lambda+2 \mu)}+o(\omega), \quad \xi_{2 n}=\frac{1}{2}+\mathcal{O}\left(\omega^{2}\right).
\]
These conclusions recover the results concerning the spectrum of the N-P operator in the static regime (cf. \cite{AJ, LiLiu2d} ).

\end{rem}

\section{Elastic resonances for material structures with no core}

In this section, we construct a broad class of elastic structures of the form $\mathbf{C}_0$ in \eqref{eq:pa1} with no core, namely $D=\emptyset$ that can induce resonances. All the notation below are carried over from 
Sections \ref{sec:introduction} and \ref{sec:auxiliary}.
Suppose that a source term $\bff$ is supported outside $\Omega$. In such a case, the elastic system \eqref{eq:lame1} can be reduced into the following transmission problem:
\begin{equation}\label{eq:nocore}
  \left\{
    \begin{array}{ll}
      \mathcal{L}_{\hat{\lambda}, \hat{\mu}}\bu(\bx) + \omega^2\bu(\bx) =0,    & \bx\in \Omega, \smallskip \\
      \mathcal{L}_{\lambda, \mu}\bu(\bx) + \omega^2\bu(\bx) =\bff, & \bx\in \mathbb{R}^N\backslash \overline{\Omega},\smallskip \\
      \bu(\bx)|_- = \bu(\bx)|_+,      & \bx\in\partial \Omega, \smallskip \\
      \partial_{\hat{\bnu}}\bu(\bx)|_- = \partial_{\bnu}\bu(\bx)|_+, & \bx\in\partial \Omega\,.
    \end{array}
  \right.
\end{equation}


\subsection{Existence of resonances in generic scenarios}
Using the single-layer potential in \eqref{eq:single}, the solution to this system can be written as 
\begin{equation}\label{eq:sol}
  \bu=
 \left\{
   \begin{array}{ll}
     \hat{\bS}_{\partial\Omega}^{\omega}[\bpsi_1](\bx), & \bx\in \Omega, \\
     \bS_{\partial\Omega}^{\omega}[\bpsi_2](\bx) + \mathbf{F}, &  \bx\in \mathbb{R}^N\backslash \overline{\Omega},
   \end{array}
 \right.
\end{equation}
where $\mathbf{F}$ is called the Newtonian potential of the source 
$\bff$ and $\bpsi_1, \bpsi_2\in L^2(\partial\Omega)^N$:
\begin{equation}\label{eq:newpf}
 \mathbf{F}(\bx):= \int_{\mathbb{R}^N} \mathbf{\Gamma}^{\omega}(\bx-\by)\bff(\by)d\by, \quad \bx\in \mathbb{R}^N\,.
\end{equation}

One can readily verify that the solution defined in \eqref{eq:sol} satisfy the first two conditions in \eqref{eq:nocore}. For the third and forth condition in \eqref{eq:nocore} across $\partial\Omega$, namely the transmission conditions, one can obtain that 
\begin{equation}\label{eq:sol1}
  \left\{
    \begin{array}{ll}
      \hat{\bS}_{\partial\Omega}^{\omega}[\bpsi_1] - \bS_{\partial\Omega}^{\omega}[\bpsi_2] = \mathbf{F}, \\
      \partial_{\hat{\bnu}}\hat{\bS}_{\partial\Omega}^{\omega}[\bpsi_1]|_- - \partial_{\bnu}\bS_{\partial\Omega}^{\omega}[\bpsi_2]|_+ = \partial_{\bnu}\mathbf{F} ,
    \end{array}
  \right.
  \quad \bx\in\partial \Omega.
\end{equation}
With the help of the jump formula \eqref{eq:jump}, the equation \eqref{eq:sol1} can be rewritten as 
\begin{equation}\label{eq:ma1}
  \bA^{\omega}
 \left[
   \begin{array}{c}
     \bpsi_1 \\
     \bpsi_2 \\
   \end{array}
 \right]=
\left[
  \begin{array}{c}
    \mathbf{F} \\
    \partial_{\bnu}\mathbf{F} \\
  \end{array}
\right],
\end{equation}
where 
\begin{equation}
  \bA^{\omega} =
 \left[
   \begin{array}{cc}
      \hat{\bS}_{\partial\Omega}^{\omega} & - \bS_{\partial\Omega}^{\omega}\smallskip \\
     -1/2I+\hat{\bK}_{\partial\Omega}^{\omega, *} & -1/2I- \bK_{\partial\Omega}^{\omega, *} \\ 
   \end{array}
 \right].
\end{equation}

Next, we show that the resonance could occur {even for the domain $\Omega$ to be of a generic geometry.}
It is noted that $\hat{\bS}_{\partial\Omega}^{\omega}$ and $\bS_{\partial\Omega}^{\omega}$ are compact operators on $L^2(\partial\Omega)^N$ (cf. \cite{bk:ab}). Following the similar argument as that in the proof of Lemma \ref{lem:sks}, one can readily show that the spectrum of the operator $\bA^{\omega}$ consist of the point spectrum only. 
{Denoting by $\mathcal{H}_j$ the generalized eigenspace of $\bA^{\omega}$ for the eigenvalue $\xi_j$, then we 
can obtain the following result, by applying the Jordan theory directly to the operator $\bA_{\delta}^{\omega}|_{\mathcal{H}_j}: {\mathcal{H}_j} \rightarrow {\mathcal{H}_j}$.}
%
\begin{lem}\label{lem:ja}
There exists a basis $\{\bm{\Psi}_{j,l,k}\}$, $1\leq l\leq m_j$, $1\leq k \leq n_{j,l}$ for $\mathcal{H}_j$ such that 
\[
\bA^{\omega} \left(\bm{\Psi}_{j,1,1}, \dots, \bm{\Psi}_{j,m_j,n_{j,m_j}}\right) =  \left(\bm{\Psi}_{j,1,1}, \dots, \bm{\Psi}_{j,m_j,n_{j,m_j}}\right)  \left(
   \begin{array}{ccc}
      J_{j,1} &   & \\
        &  \ddots  &\\
          &  & J_{j,m_j} \\
   \end{array}
 \right),
\]
where $J_{j,l}$ is the canonical Jordan matrix of size $n_{j,l}$ in the form 
\[
 J_{j,l}=\left(
   \begin{array}{cccc}
      \xi_j & 1  & & \\
        &  \ddots  & \ddots & \\
          &  & \xi_j &1 \\
          & & & \xi_j
   \end{array}
 \right).
\]
\end{lem}

The following theorem presents the existence of resonances in generic scenarios. 

\begin{thm}\label{th:reg}
If the parameters are properly chosen such that one of the eigenvalues in Lemma~\ref{lem:ja} satisfies $\xi_j\ll1$, then resonance occurs in the sense of Definition~\ref{def:1}. Denote by $p_{j}=\max\{n_{j,l} \}_{l=1}^{m_j}$ with $n_{j,l}$ defined in Lemma \ref{lem:ja}, then the resonance occurs at the order of $M=\mathcal{O}\left({\xi_j^{-2p_j}} \right)$ (cf. \eqref{con:res}).
\end{thm}
\begin{proof}
Without loss of generality, we assume that $p_j=n_{j,1}$ and consider the source term only containing the following term
\[
 \left[
  \begin{array}{c}
    \mathbf{F} \\
    \partial_{\bnu}\mathbf{F} \\
  \end{array}
\right]=
\sum_{k=1}^{p_j} f_k \bm{\Psi}_{j,1,k},
\]
where $f_k$ are the coefficients. Thus the density function can be written as
\[
 \left[
   \begin{array}{c}
     \bpsi_1 \\
     \bpsi_2 \\
   \end{array}
 \right]=
 \sum_{k=1}^{p_j} g_k \bm{\Psi}_{j,1,k},
\]
where $g_k$ are the coefficients to be determined. With the help of \eqref{eq:ma1} and Lemma \ref{lem:ja}, one has that 
\[
g_k = \frac{1}{\xi_j^{p_j-k+1}}\sum_{i=k}^{p_j}(f_{i}\xi_j^{p_j-i} (-1)^{i-k+1}).
\]
Thus $g_1$ has the following expression
\[
 g_1= \frac{f_{p_j} }{\xi_j^{p_j}} (-1)^{p_j} + \mathcal{O}(\xi_j^{1-p_j}).
\]
{Then we have the following estimate for the dissipation energy:} 
\[
\begin{split}
E(\bu) &=  \Im P_{\hat{\lambda},\hat{\mu}}(\bu,\bu)  = \Im \left( \omega^2\int_{\Omega} |\bu|^2 d\bx + \int_{\partial \Omega} \bu \cdot \partial_{\bnu} \overline{\bu}  \right) \\
 & \geq |g_1|^2 \Im \left(  \int_{\partial \Omega} \hat{\bS}_{\partial\Omega}^{\omega}[ \bm{\Psi}_{j,1,1}] \cdot \overline{ \partial_{\hat{\bnu}} \hat{\bS}_{\partial\Omega}^{\omega}[ \bm{\Psi}_{j,1,1}]} |_- \right)\\
 & \geq  \frac{f_{p_j}^2 }{\xi_j^{2p_j}}  \Im \left(  \int_{\partial \Omega} \hat{\bS}_{\partial\Omega}^{\omega}[ \bm{\Psi}_{j,1,1}] \cdot \overline{ \partial_{\hat{\bnu}} \hat{\bS}_{\partial\Omega}^{\omega}[ \bm{\Psi}_{j,1,1}]} |_- \right).
\end{split}
\]
This completes the proof.

\end{proof}

\begin{rem}
We would like to emphasize that the resonance within the finite regime is highly enhanced compared with that in the static case. In fact, in the static situation, the N-P operator $\hat{\bK}_{\partial\Omega}^{0, *}$ is symmetric in a certain Hilbert space. Thus the parameter $p_j$ in Theorem \ref{th:reg} is always $1$ and the resonance always blows up at the rate $\mathcal{O}\left({\xi_j^{-2}} \right)$ \cite{AJ}. However, {in our current case} beyond the quasi-static approximation, the resonance would blow up at the rate $\mathcal{O}\left({\xi_j^{-2p_j}} \right)$ as proved in the last theorem with $p_j\geq 1$, which is caused by the fact that the N-P operator $\hat{\bK}_{\partial\Omega}^{\omega, *}$ is no longer symmetric 
in any inner product space. 
\end{rem}

\begin{rem}
The condition $\xi_j\ll1$ generally can be satisfied. In fact, since the Lam\'e parameters $(\hat{\lambda}, \hat{\mu})$ in the domain $\Omega$ can break the strong convexity conditions in \eqref{eq:con}, hence the system \eqref{eq:nocore} is allowed to lose the ellipticity. Thus there exists a certain eigenvalue satisfying the condition $\xi_j\ll1$. Next, we choose the domain $\Omega$ to be a circle to strictly verify the statement in Theorem \ref{th:reg} in two dimensions. For the three dimensions, readers may refer to the paper \cite{DLL2}.
\end{rem}

\subsection{Resonance and its quantitative behavior for circular domain}
In this subsection, we consider the specific case that the domain $\Omega$ is a circle $B_R$. In such a case, we can have a deep understanding of the occurrence of the resonance as well as its quantitative behaviours. Since the source term $\bff$ is supported outside $B_R$, there exists $\epsilon>0$ such that when $\bx\in B_{R+\epsilon}$, the Newtonian potential $\bF$ defined in \eqref{eq:newpf} satisfies 
\[
 \Lcal_{\lambda,\mu}\bF + \omega^2\bF=0.
\]
Thus $\bF$ can be written as 
\begin{equation}\label{eq:FF}
  \bF= \sum_{n=-\infty}^{\infty} \left(  \frac{ \kappa_{1,n} k_s R}{n J_n(k_s R)}   \mathbf{Q}_n^i  +  \frac{ \kappa_{2,n} k_p R}{n J_n(k_p R)} \mathbf{P}_n^i \right),
\end{equation}
where $\kappa_{1,n}, \kappa_{2,n}$ are the coefficients, and the functions $\mathbf{Q}_n^i $ and $\mathbf{P}_n^i$ are defined in \eqref{eq:qi} and \eqref{eq:pi}. Here $\frac{ k_s R}{n J_n(k_s R)}$ and $\frac{ k_p R}{n J_n(k_p R)}$ are the normalization constants. From the expressions for the functions $\mathbf{Q}_n^i $ and $\mathbf{P}_n^i$ in \eqref{eq:qi} and \eqref{eq:pi}, one has that on the boundary $\partial B_R$:
\[
  \bF= \sum_{n=-\infty}^{\infty}  \bb_n^t \bff_n,
\]
where 
\[
\bb_n= \left(
   \begin{array}{c}
     e^{\rmi n \theta} \bnu \\
     e^{\rmi n \theta} \bt \\
   \end{array}
 \right), \quad 
 \bff_n =  \left(
   \begin{array}{c}
    f_{1,n}  \\
    f_{2,n} \\
   \end{array}
 \right)=  \left(
   \begin{array}{c}
     \kappa_{1,n} \eta_{1,n} +  \kappa_{2,n} \eta_{3,n}  \\
     \kappa_{1,n} \eta_{2,n} +  \kappa_{2,n} \eta_{4,n} \\
   \end{array}
 \right),
\]
with 
\[
\eta_{1,n}=2, \quad \eta_{2,n}= \frac{2\rmi k_s R J_{n}^{\prime}(k_s R)}{n J_n(k_s R)}, \quad  \eta_{3,n}=   \frac{ 2k_p R J_{n}^{\prime}(k_p R)}{n J_n(k_p R)}  , \quad \eta_{4,n}=   2 \rmi .
\]
Moreover, from \eqref{eq:tpi}, one has that on the boundary $\partial B_R$:
\[
 \partial_{\bnu} \bF= \sum_{n=-\infty}^{\infty}  \bb_n^t \tilde{\bff}_n,
\]
where
\[
 \tilde{\bff}_n =  \left(
   \begin{array}{c}
    \tilde{f}_{1,n}  \\
    \tilde{f}_{2,n} \\
   \end{array}
 \right)= \left(
   \begin{array}{c}
     \frac{  \kappa_{1,n} \gamma_{1,n}  k_s R}{n J_n(k_s R)} +  \frac{ \kappa_{2,n} \gamma_{3,n} k_p R}{n J_n(k_p R)}  \\
    \frac{   \kappa_{1,n} \gamma_{2,n}  k_s R}{n J_n(k_s R)} + \frac{  \kappa_{2,n} \gamma_{4,n} k_p R}{n J_n(k_p R)} \\
   \end{array}
 \right),
\]
with $\gamma_{i,n}$, $1\leq i \leq 4$ given in \eqref{eq:tpi}.

From Lemmas \ref{lem:eisin} and \ref{lem:pnu}, one has that under the basis $\left(e^{\rmi n \theta}\bnu, e^{\rmi n \theta}\bt \right)$, the operators $ \bS_{\partial\Omega}^{\omega}, \hat{\bS}_{\partial\Omega}^{\omega}, \partial_{\bnu}\bS_{\partial\Omega}^{\omega}[\bpsi_2]|_+, \partial_{\hat{\bnu}}\hat{\bS}_{\partial\Omega}^{\omega}[\bpsi_1]|_-$ have the following expressions:
\[
 \bS_{\partial\Omega}^{\omega}= \Tcal_{1n}, \quad  \hat{\bS}_{\partial\Omega}^{\omega} =  \widehat{\Tcal}_{1n}, \quad \partial_{\bnu}\bS_{\partial\Omega}^{\omega}[\bpsi_2]|_+= \Tcal_{2n},\quad  \partial_{\hat{\bnu}}\hat{\bS}_{\partial\Omega}^{\omega}[\bpsi_1]|_- = \widehat{\Tcal}_{2n},
\]
where 
\[
{\Tcal}_{1n} =
 \left(
   \begin{array}{cc}
      \alpha_{1, n} &  \alpha_{3, n} \\
      \alpha_{2, n} &  \alpha_{4, n} \\
   \end{array}
 \right), \quad
 \widehat{\Tcal}_{1n} =
 \left(
   \begin{array}{cc}
      \hat{\alpha}_{1, n} &  \hat{\alpha}_{3, n} \\
      \hat{\alpha}_{2, n} &  \hat{\alpha}_{4, n} \\
   \end{array}
 \right),  
\]
\[
{\Tcal}_{2n} =
 \left(
   \begin{array}{cc}
     g_{1, n} & g_{3, n} \\
     g_{2, n} &  g_{4, n} \\
   \end{array}
 \right), \quad
 \widehat{\Tcal}_{2n} =
 \left(
   \begin{array}{cc}
     \hat{g}_{1, n}-1 & \hat{g}_{3, n} \\
     \hat{g}_{2, n} & \hat{g}_{4, n}-1 \\
   \end{array}
 \right).  
\]
In the last equation, $\alpha_{in}$ and $g_{i,n}$ with $i=1,2,3,4$ are given in Lemmas \ref{lem:eisin} and \ref{lem:pnu}, and $\hat{\alpha}_{in}$ and $\hat{g}_{i,n}$ with $i=1,2,3,4$ are also given in Lemmas \ref{lem:eisin} and \ref{lem:pnu} with $(\mu, \lambda)$ replaced by $(\hat{\mu}, \hat{\lambda})$.

Hence, the density functions $\bpsi_1$ and $\bpsi_2$ can be expressed by, 
\begin{equation}\label{eq:denc}
\bpsi_1 =\sum_{n=-\infty}^{\infty} \bb_n^t \bpsi_{1,n}, \qquad  \bpsi_2 =\sum_{n=-\infty}^{\infty} \bb_n^t \bpsi_{2,n},
\end{equation}
where 
\[
\bb_n= \left(
   \begin{array}{c}
     e^{\rmi n \theta} \bnu \\
     e^{\rmi n \theta} \bt \\
   \end{array}
 \right), \quad 
  \bpsi_{1,n} = \left(
   \begin{array}{c}
      \psi_{1,1,n} \\
      \psi_{1,2,n} \\
   \end{array}
 \right), \quad 
 \bpsi_{2,n} = \left(
   \begin{array}{c}
      \psi_{2,1,n} \\
      \psi_{2,2,n} \\
   \end{array}
 \right),
\]
and the coefficients $\psi_{i,j,n}$, $1\leq i,j\leq 2$ are needed to be determined. Thus the system \eqref{eq:sol1} can be written as for $-\infty<n<\infty$:
\begin{equation}\label{eq:tm}
  \left\{
    \begin{array}{ll}
      \widehat{\Tcal}_{1n} \bpsi_{1,n} = \Tcal_{1n} \bpsi_{2,n}  + \bff_n , \\
      \widehat{\Tcal}_{2n} \bpsi_{1,n} = \Tcal_{2n} \bpsi_{2,n}  + \tilde{\bff}_n.
    \end{array}
  \right.
\end{equation}
Directly solving the equation \eqref{eq:tm} gives that
\begin{equation}\label{eq:sona}
\psi_{1,1,n} =\frac{c_{1,n}}{d_n},  \quad \psi_{1,2,n} =\frac{c_{2,n}}{d_n},
\end{equation}
where
\[
\begin{split}
c_{1,n} =& ( f_{2,n} \hat{\alpha}_{3, n} - f_{1,n} \hat{\alpha}_{4, n})( g_{3, n} g_{2, n} - g_{1, n} g_{4, n}) +
	 ( f_{2,n} (\hat{g}_{4, n}-1) - \tilde{f}_{2,n} \hat{\alpha}_{4, n})( g_{1, n} \alpha_{3, n} - g_{3, n} \alpha_{1, n}) + \\
	&  ( f_{2,n} \hat{g}_{3, n} - \tilde{f}_{1,n} \hat{\alpha}_{4, n})( g_{4, n} \alpha_{1, n} - g_{2, n} \alpha_{3, n}) +
	 ( f_{1,n} (\hat{g}_{4, n}-1) - \tilde{f}_{2,n} \hat{\alpha}_{3, n})( g_{3, n} \alpha_{2, n} - g_{1, n} \alpha_{4, n}) + \\
	 & ( f_{1,n} \hat{g}_{3, n} - \tilde{f}_{1,n} \hat{\alpha}_{3, n})( g_{2, n} \alpha_{4, n} - g_{4, n} \alpha_{2, n}) +
	 ( \tilde{f}_{1,n} (\hat{g}_{4, n}-1) - \tilde{f}_{2,n} \hat{g}_{3, n})( \alpha_{1, n} \alpha_{4, n} - \alpha_{3, n} \alpha_{2, n}) ,
\end{split}
\]
\[
\begin{split}
c_{2,n} =&( f_{2,n} \hat{\alpha}_{1, n} - f_{1,n} \hat{\alpha}_{2, n})( g_{1, n} g_{4, n} - g_{3, n} g_{2, n}) +
	 ( \tilde{f}_{2,n} (\hat{g}_{1, n}-1) - \tilde{f}_{1,n} \hat{g}_{2, n} )( \alpha_{1, n} \alpha_{4, n} -  \alpha_{3, n} \alpha_{2, n} ) + \\
	&( f_{1,n} \hat{g}_{2, n} - \tilde{f}_{2,n} \hat{\alpha}_{1, n})( g_{1, n} \alpha_{4, n} - g_{3, n} \alpha_{2, n}) + 
	 ( f_{2,n} (\hat{g}_{1, n}-1) - \tilde{f}_{1,n} \hat{\alpha}_{2, n})( g_{2, n} \alpha_{3, n} - g_{4, n} \alpha_{1, n}) +\\
	& ( f_{2,n} \hat{g}_{2, n} - \tilde{f}_{2,n} \hat{\alpha}_{2, n})( g_{3, n} \alpha_{1, n} - g_{1, n} \alpha_{3, n}) +
	 ( f_{1,n} (\hat{g}_{1, n}-1) - \tilde{f}_{1,n} \hat{\alpha}_{1, n})( g_{4, n} \alpha_{2, n} - g_{2, n} \alpha_{4, n}),
\end{split}
\]
and
\[
\begin{split}
d_n = &( \hat{\alpha}_{1, n} \hat{\alpha}_{4, n} - \hat{\alpha}_{3, n} \hat{\alpha}_{2, n})( g_{1, n} g_{4, n} - g_{3, n} g_{2, n}) +  \\
     & ( \hat{g}_{3, n} \hat{\alpha}_{2, n} - \hat{\alpha}_{4, n} (\hat{g}_{1, n}-1) )( g_{4, n} \alpha_{1, n} - g_{2, n} \alpha_{3, n}) + \\ 
	& ( \hat{g}_{3, n} \hat{\alpha}_{2, n} - \hat{\alpha}_{4, n} (\hat{g}_{1, n}-1) )( \alpha_{1, n} g_{4, n} - g_{2, n} \alpha_{3, n}) + \\
	&( \hat{\alpha}_{3, n} \hat{g}_{2, n} - \hat{\alpha}_{1, n} (\hat{g}_{4, n}-1))( g_{1, n} \alpha_{4, n} - g_{3, n} \alpha_{2, n}) +\\
	&(  \hat{\alpha}_{1, n} \hat{g}_{3, n} - \hat{\alpha}_{3, n} (\hat{g}_{1, n}-1) )( g_{2, n} \alpha_{4, n} -  g_{4, n} \alpha_{2, n}) + \\
	&( \hat{g}_{2, n} \hat{g}_{3, n} (\hat{g}_{4, n}-1) (\hat{g}_{1, n}-1) )( \alpha_{3, n} \alpha_{2, n} - \alpha_{1, n} \alpha_{4, n}) .
\end{split}
\]

\begin{thm}\label{thm:reson}
Consider the configuration $\mathbf{C}_0$ with $D=\emptyset$ defined in \eqref{eq:pa1} and a source term $\bff$ supported outside the domain $\Omega$.
If the Lam\'e parameters $(\hat{\lambda}, \hat{\mu})$ inside the domain $\Omega$ is chosen such that for any $M\in\mathbb{R}_+$:
\begin{equation}\label{eq:renoc}
 |\psi_{1,1,n_0}| >M,
\end{equation}
for some $n_0\in\mathbb{N}$, where $\psi_{1,1,n_0}$ is defined in \eqref{eq:sona}, then the elastic resonance occurs. 
\end{thm}
\begin{proof}
With the help of the Green's formula, the dissipation energy defined in \eqref{def:E} can be written as 
\[
\begin{split}
E(\bu) &=  \Im P_{\hat{\lambda},\hat{\mu}}(\bu,\bu)  = \Im \left( \omega^2\int_{\Omega} |\bu|^2 d\bx + \int_{\partial \Omega} \bu \cdot \partial_{\bnu} \overline{\bu}  \right)  = \Im \left(  \int_{\partial \Omega} \bu \cdot \partial_{\bnu}\overline{\bu} \right) \\
 & \geq |\psi_{1,1,n}|^2 \Im \left(  \int_{\partial \Omega} \hat{\bS}_{\partial\Omega}^{\omega}[e^{\rmi n \theta} \bnu] \cdot \overline{ \partial_{\hat{\bnu}} \hat{\bS}_{\partial\Omega}^{\omega}[e^{\rmi n \theta} \bnu]} |_- \right),
\end{split}
\]
which shows that the resonance occurs thanks to \eqref{eq:renoc} and completes the proof. 
\end{proof}

\begin{rem}
If the Lam\'e parameters $(\hat{\lambda}, \hat{\mu})$ inside the domain $\Omega$ are chosen as follows:
\begin{equation}\label{eq:spa}
(\hat{\lambda}, \hat{\mu}) = c (\lambda, \mu),
\end{equation}
where $(\lambda, \mu)$ are the Lam\'e parameters in $\mathbb{R}^2\backslash \overline{\Omega}$. For a large order $n$ such that the asymptotic expansions \eqref{eq:asn} hold, the parameter $c$ should have the following asymptotic expansion such that the condition \eqref{eq:renoc} holds:
\[
 c=-\frac{\lambda + 3\mu}{\lambda + \mu} + \vartheta_n, 
\] 
where $\vartheta_n =\Ocal(1/n)$. In fact, for a large order $n$, the solutions of the equation \eqref{eq:tm} have the following asymptotic expansions:
\begin{equation*}\label{eq:so1}
\begin{split}
\psi_{1,1,n} =\frac{\varepsilon_{1,n}}{((1+c)\lambda + (3+c)\mu) \varrho +  \Ocal(1/n)}, \ \
\psi_{1,2,n} =\frac{\varepsilon_{2,n}}{((1+c)\lambda + (3+c)\mu) \varrho + \Ocal(1/n)},
\end{split}
\end{equation*}
where $\varrho$ is a constant not depending on $n$ and
\[
\begin{split}
\varepsilon_{1,n} = &  \left((c-1) \omega^2 R^2 (\lambda+\mu) (\lambda+3 \mu)-8 c \mu (\lambda+2 \mu) ((c+3) \lambda+(c+7) \mu)\right)\times \\
& 256 c \mu^2 n (c_1+c_2) (\lambda+2 \mu)^3\left( 1+ \Ocal\left(\frac{1}{n}\right) \right),
\end{split}
\]
\[
\begin{split}
\varepsilon_{2,n} = & -\left((c-1) \omega^2 R^2 (\lambda+\mu) (\lambda+3 \mu)+8 c \mu (\lambda+2 \mu) (3 c \lambda+(3 c+5) \mu+\lambda)\right)\times \\
& 256 \rmi c \mu^2 n (c_1+c_2) (\lambda+2 \mu)^3\left( 1+ \Ocal\left(\frac{1}{n}\right) \right).
\end{split}
\]
Thus one can readily conclude that the parameter $c$ should have the following asymptotic expansion 
\[
 c=-\frac{\lambda + 3\mu}{\lambda + \mu} + \Ocal\left( \frac{1}{n} \right), 
\] 
such that the condition \eqref{eq:renoc} holds.
\end{rem}

Next, we show that the condition \eqref{eq:renoc} can be achieved. The Lam\'e parameters inside the domain $\Omega$ are chosen as those in \eqref{eq:spa}. The other parameters are chosen as follows:
\[
 n=5, \quad \lambda=\mu=\omega=R=1, \quad \Im c=2.08\times10^{-9}.
\]
This is the case beyond the quasi-static approximation from the values of $\omega$ and $R$. The absolute value of $\psi_{1,1,n}$ given in \eqref{eq:sona} with respect to the real part of $c$, i.e. $\Re c$, is plotted in Fig.\ref{fig:reno1}. This clearly shows that the condition \eqref{eq:renoc} is fulfilled and thus the resonance occurs.
\begin{figure}[t]
\centering
\includegraphics[width=0.3\textwidth]{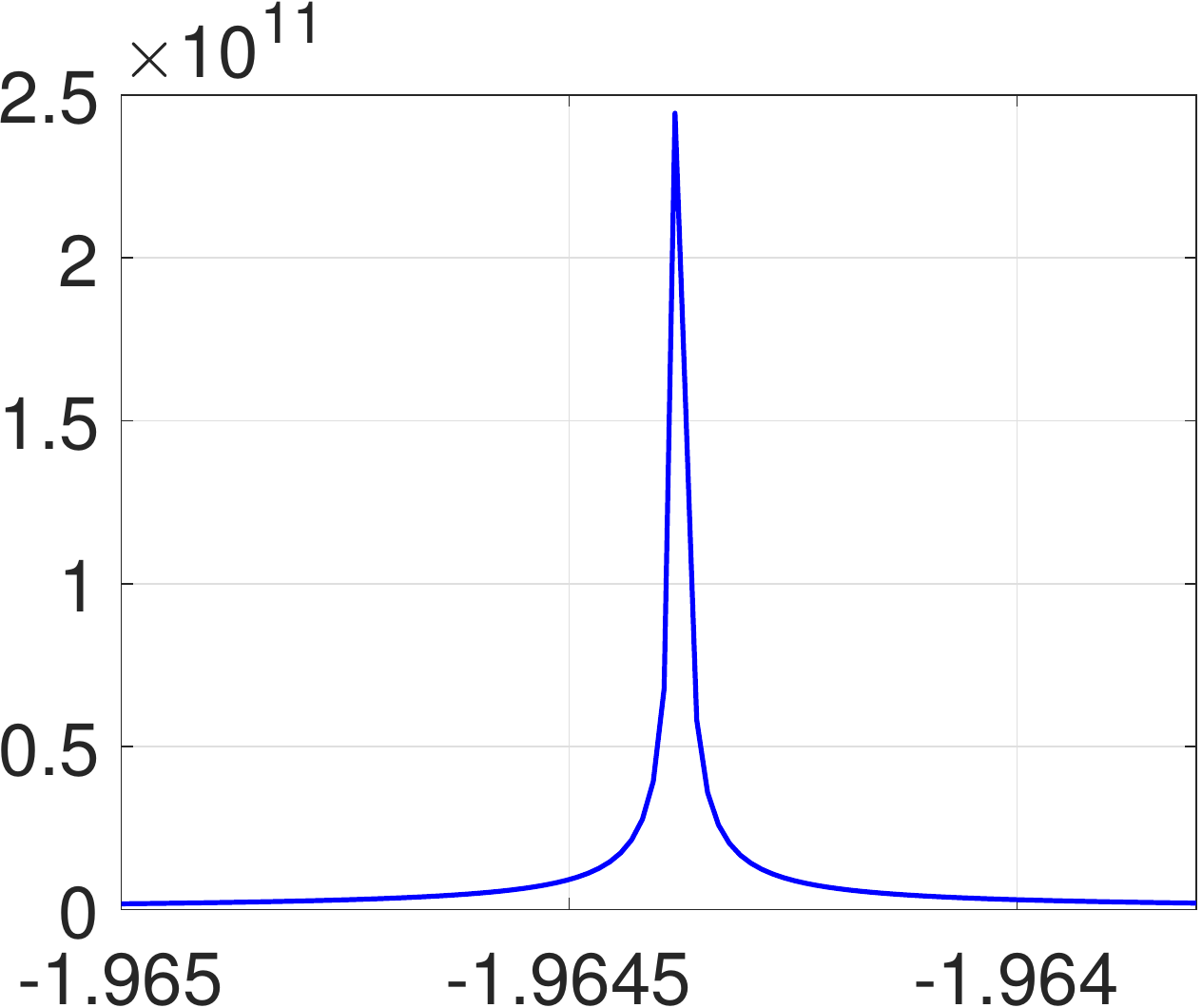}
\caption{\label{fig:reno1}  The absolute value of $\psi_{1,1,n}$ given in \eqref{eq:sona} with respect to $\Re c$. Horizontal axis: value of $\Re c$; Vertical axis: absolute value of $\psi_{1,1,n}$.}
\end{figure}

\begin{rem}\label{rem:nn1}
To ensure the occurrence of the resonance, in the quasi-static case, the condition $\Im c\rightarrow 0$ is required (cf. \cite{DLL, LLL}). However, in our current case beyond the quasi-static regime, one usually requires $\Im c\rightarrow c^*$ with $c^*\neq 0$. This is a sharp difference from the quasi-static case. Next, we conduct a numerical simulation to verify this statement. The parameters are chosen as follows
\[
 n=5, \quad \lambda=\mu=\omega=R=1, \quad \Re c=-1.9643,
\]
which is the case beyond the quasi-static approximation from the values of $\omega$ and $R$. The absolute value of $\psi_{1,1,n}$ given in \eqref{eq:sona} with respect to the imaginary part of $c$, i.e. $\Im c$, is plotted in Fig.\ref{fig:reno2}. This clearly shows that the resonance occurs and the critical value $\Im c\neq 0$.
\begin{figure}[t]
\centering
\includegraphics[width=0.3\textwidth]{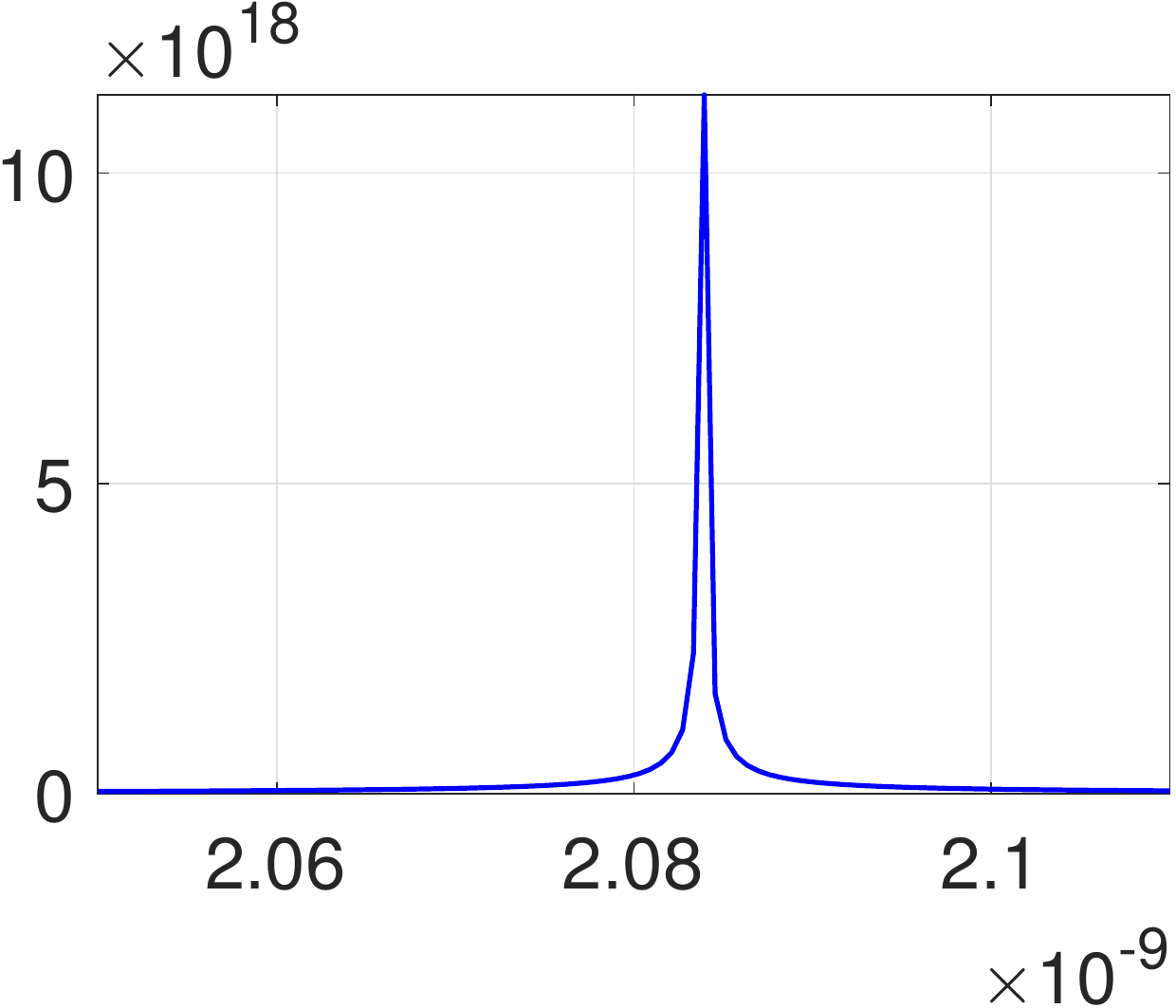}
\caption{\label{fig:reno2}  The absolute value of $\psi_{1,1,n}$ given in \eqref{eq:sona} with respect to $\Im c$. Horizontal axis: value of $\Im c$; Vertical axis: absolute value of $\psi_{1,1,n}$.}
\end{figure}
\end{rem}

Finally, we consider the quantitative behaviours of the resonant fields when resonance occurs. It is recalled that in the static/quasi-static regime, the plasmon/polariton resonances are localized around the metamaterial interface. However, we shall show that in the frequency regime beyond the quasi-static approximation, the resonant oscillation outside the material structure is localized around the metamaterial interface, but inside the material structure it is not localized around the interface, which is in sharp contrast to the subwavelength resonance. In fact, from the expression of the solution in \eqref{eq:sol} and the density functions in \eqref{eq:denc}, it is sufficient to analyze the properties of single layer potentials  $\mathbf{S}_{\partial B_{R}}^{\omega}\left[e^{\mathrm{i} n \theta} \boldsymbol{\nu}\right](\bx)$ and $\mathbf{S}_{\partial B_{R}}^{\omega}\left[e^{\mathrm{i} n \theta} \mathbf{t}\right](\bx)$ expressed in Theorem \ref{thm:sing} and Proposition \ref{pro:sinin} for $\bx$ lying in different regions. Here, we only take the term  $\mathbf{S}_{\partial B_{R}}^{\omega}\left[e^{\mathrm{i} n \theta} \boldsymbol{\nu}\right](\bx)$ to illustrate the phenomenon {as the discussion 
is the same for the term $\mathbf{S}_{\partial B_{R}}^{\omega}\left[e^{\mathrm{i} n \theta} \mathbf{t}\right](\bx)$.} The parameters are chosen as follows:
\begin{equation}\label{eq:pa1_n}
n=5, \quad \lambda=\mu=R=1, \quad \omega=20,
\end{equation}
which is the case beyond the quasi-static approximation from the values of $\omega$ and $R$. The amplitude of the single layer potential $\mathbf{S}_{\partial B_{R}}^{\omega}\left[e^{\mathrm{i} n \theta} \boldsymbol{\nu}\right](\bx)$ for $|\bx|\leq 1$ and $1<|\bx|<3$ are plotted in Fig.\,\ref{sin1}(a) and (b), respectively. From the plot, one can conclude that the field outside $B_1$ is localized around the surface $\partial B_1$, while the field inside $B_1$ is not localized around the boundary. 
\begin{figure}[t]
\centering
\subfigure[]{
\includegraphics[width=0.4\textwidth]{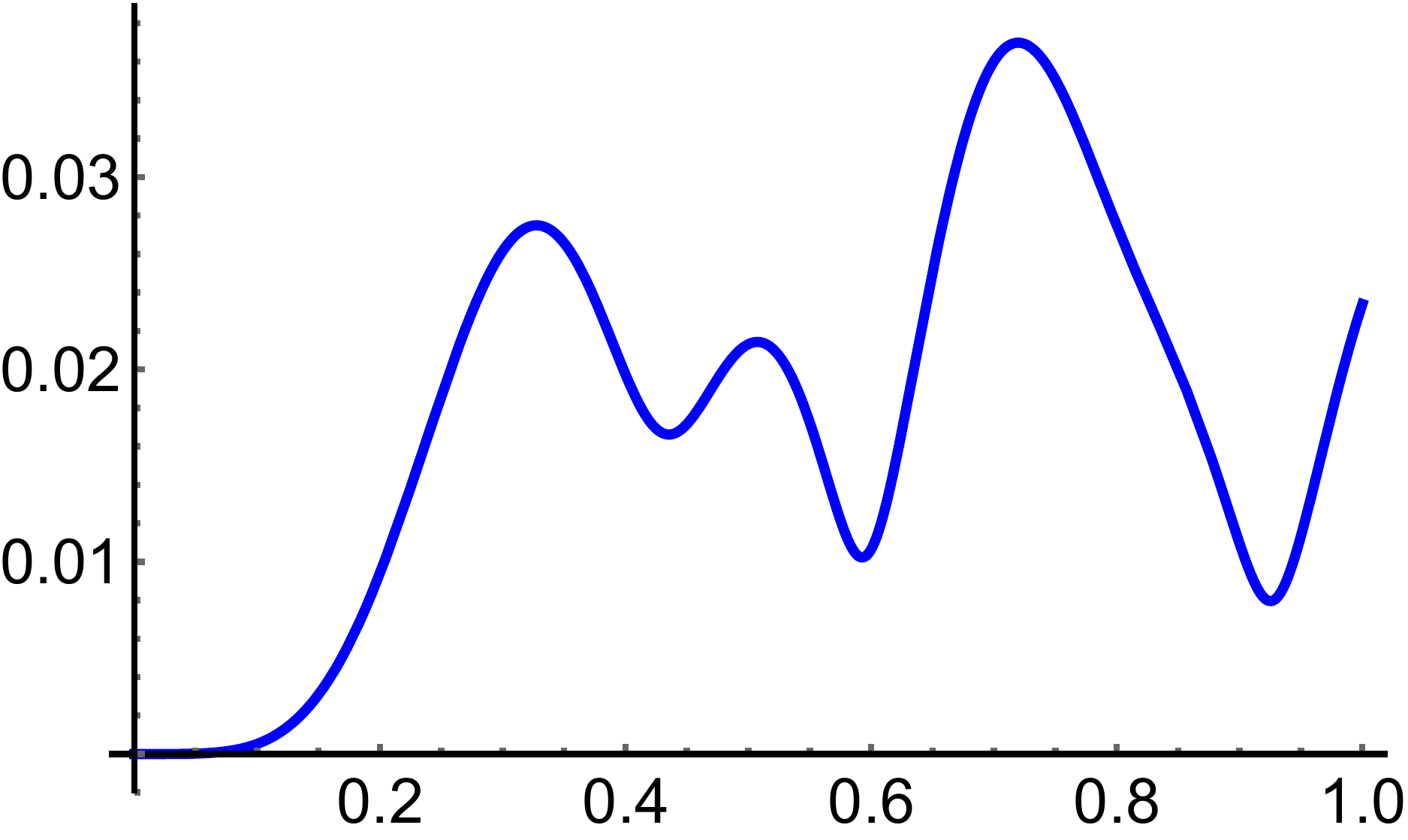}}
\hspace{1cm}
\subfigure[]{
\includegraphics[width=0.4\textwidth]{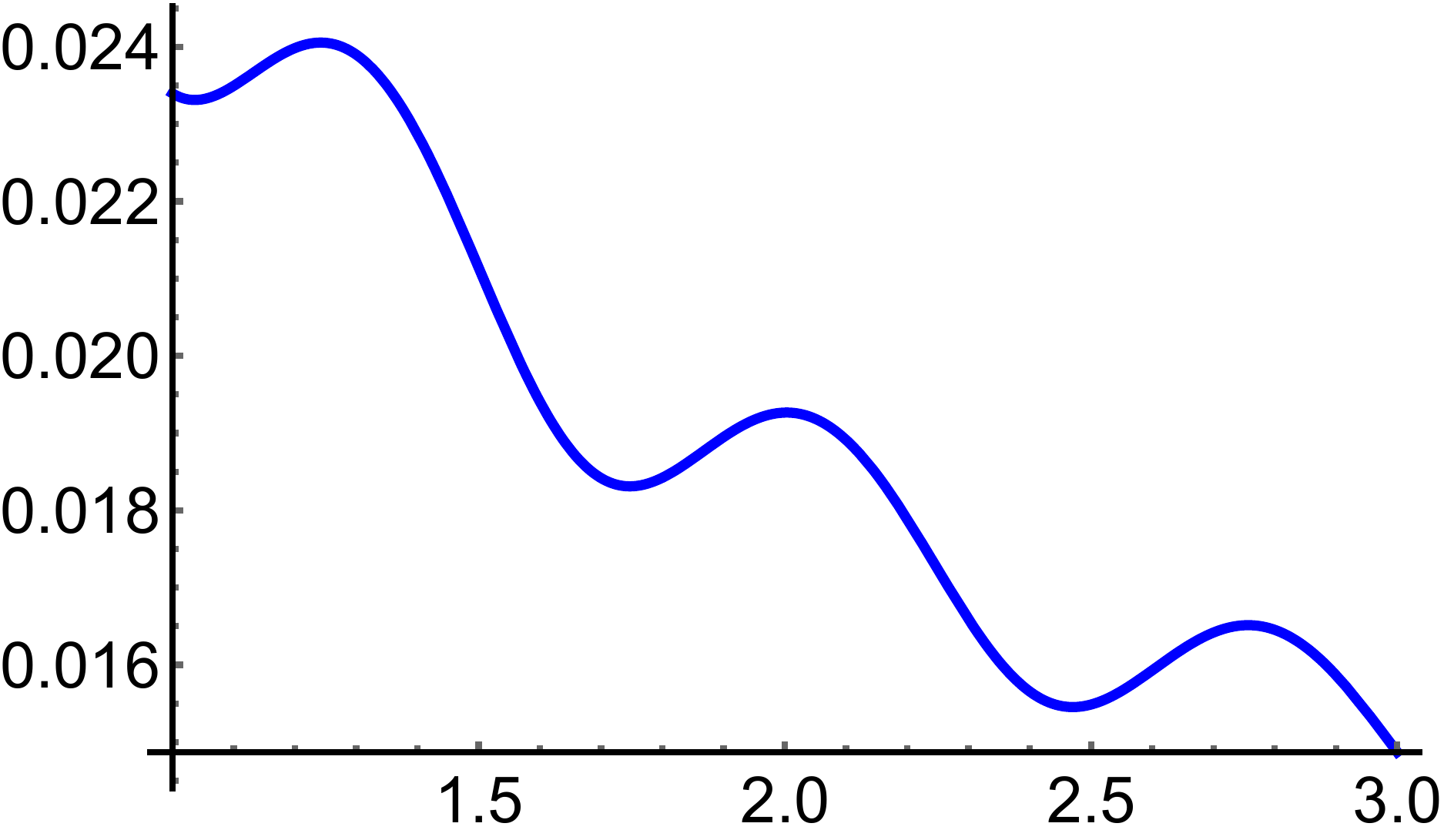}}\\
\caption{\label{sin1} The amplitude of the single layer potential $\mathbf{S}_{\partial B_{R}}^{\omega}\left[e^{\mathrm{i} n \theta} \boldsymbol{\nu}\right](\bx)$ with parameters chosen in \eqref{eq:pa1_n} for (a) ~$|\bx|\leq 1$; (b) ~$1<|\bx|\leq 3$.  }
\end{figure}
If we choose the parameters as follows:
\begin{equation}\label{eq:pa2}
 n=5, \quad \lambda=\mu=R=1, \quad \omega=0.1,
\end{equation}
which is the case of the quasi-static approximation. The amplitude of the single layer potential $\mathbf{S}_{\partial B_{R}}^{\omega}\left[e^{\mathrm{i} n \theta} \boldsymbol{\nu}\right](\bx)$ for $|\bx|\leq 1$ and $1<|\bx|<2$ are plotted in 
Fig.\,\ref{sin2}(a) and (b), respectively. From the plot, one can conclude that the fields both inside and outside $B_1$ are localized around the surface $\partial B_1$. Finally, we would like to remark that by using the relevant results in \cite{DLL}, one can show that the elastodynamical resonances in 3D reveal similar behaviours as the 2D case discussed above. 
\begin{figure}[t]
\centering
\subfigure[]{
\includegraphics[width=0.4\textwidth]{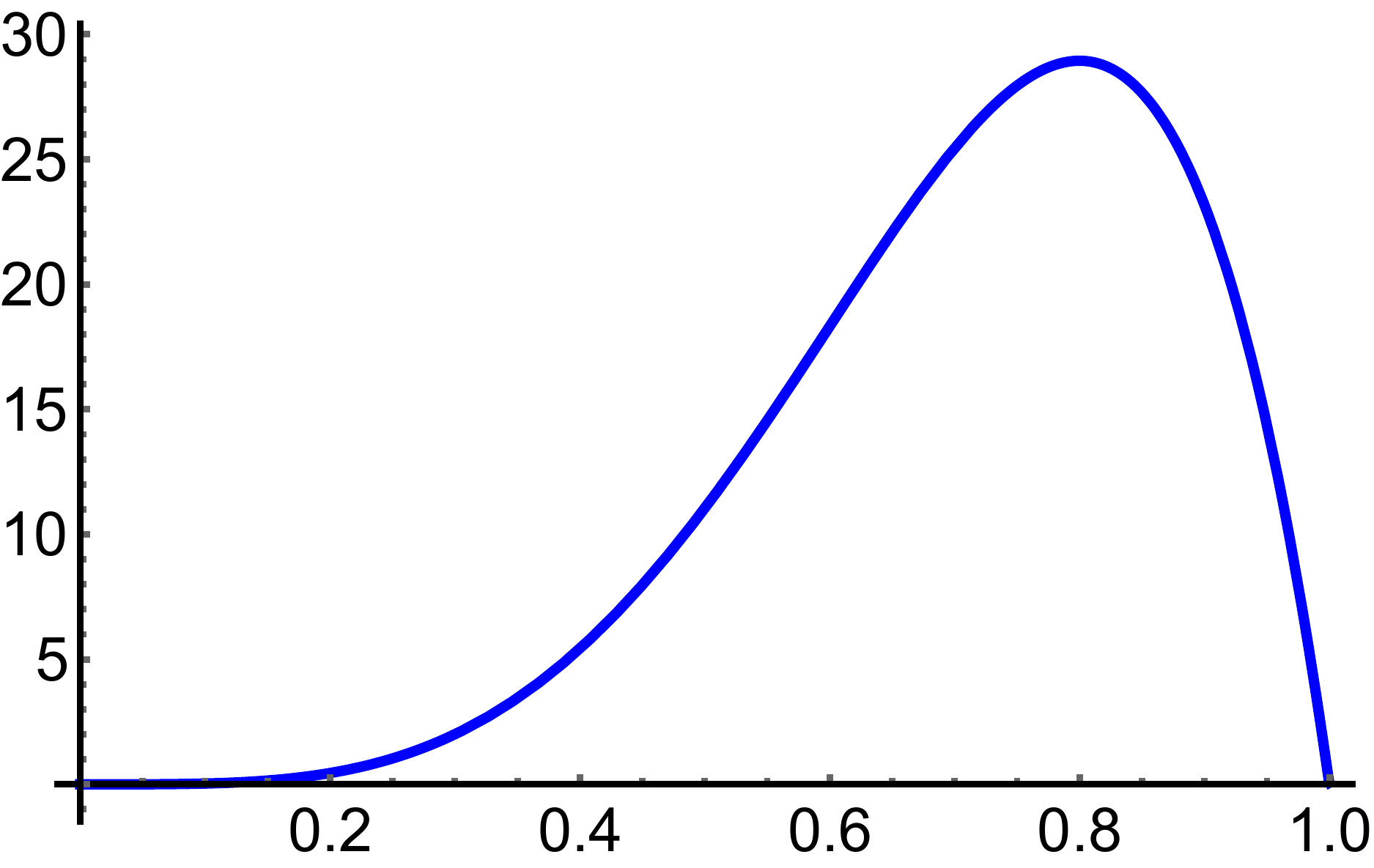}}
\hspace{1cm}
\subfigure[]{
\includegraphics[width=0.4\textwidth]{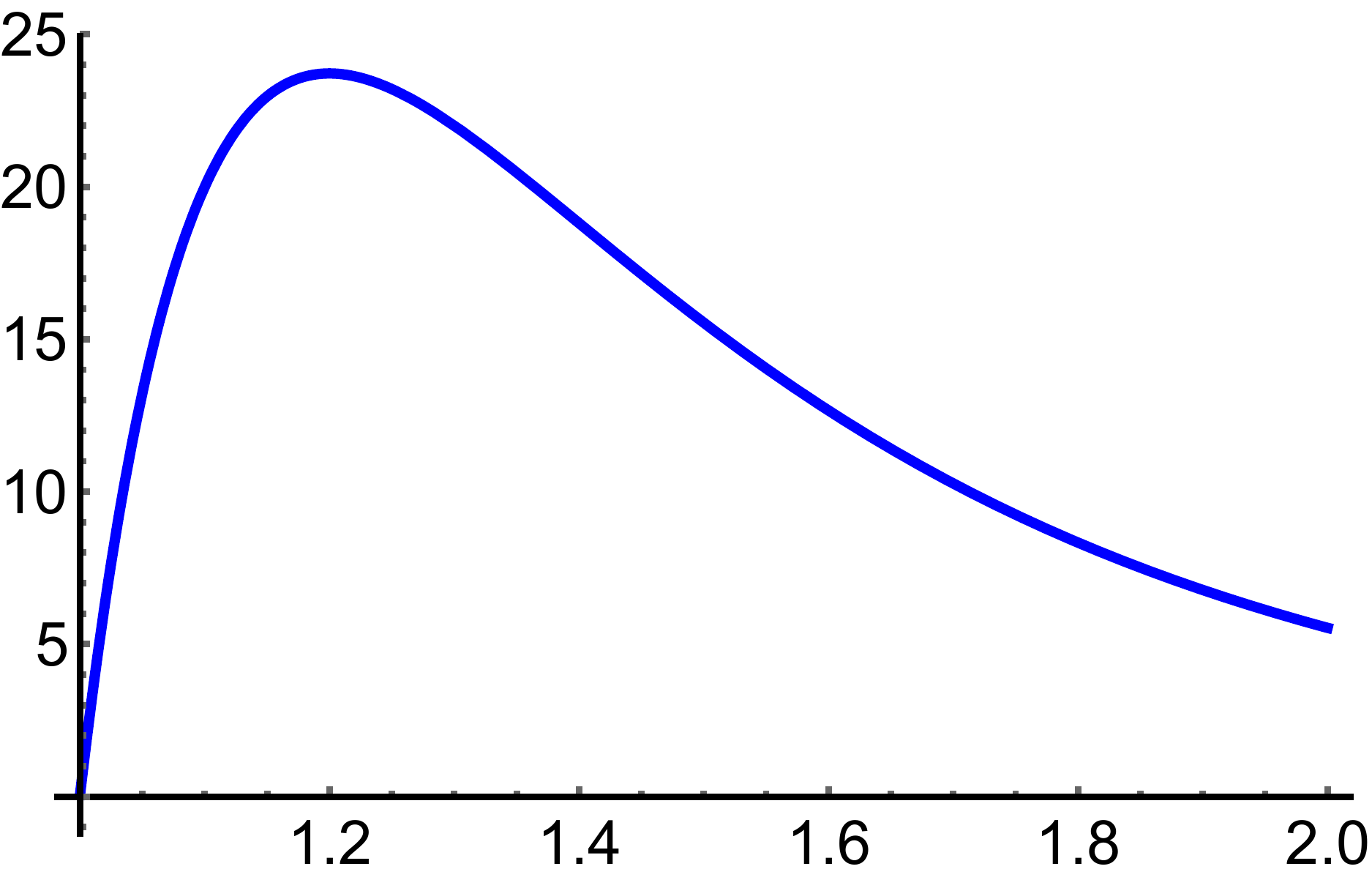}}\\
\caption{\label{sin2} The amplitude of the single layer potential $\mathbf{S}_{\partial B_{R}}^{\omega}\left[e^{\mathrm{i} n \theta} \boldsymbol{\nu}\right](\bx)$ with parameters chosen in \eqref{eq:pa2} for (a) ~$|\bx|\leq 1$; (b) ~$1<|\bx|\leq 2$.  }
\end{figure}

 \section{CALR for a core-shell structure beyond the quasi-static approximation}
 
In this section, we construct a core-shell elastic structure that can induce anomalous localized resonance; see Definition~\ref{def:1}. We confine our study in two dimensions and as mentioned earlier, we refer to \cite{DLL2} for related studies in the three-dimensional case. In what follows, we let $D=B_{r_i}$ and $\Omega=B_{r_e}$, $r_e>r_i$.  Moreover, we let $\mathcal{L}_{\breve{\lambda}, \breve{\mu}}$, $\partial_{\breve{\bnu}}$, $\breve{\bS}_{\partial D}$ and $(\breve{\bK}_{\partial D}^{\omega})^* $, respectively, denote the Lam\'e operator, the associated conormal derivative, the single layer potential operator and the N-P operator associated with the Lam\'e parameters $(\breve{\lambda}, \breve{\mu})$. 

 Assume that the source $\bff$ is supported outside $\Omega$. Associated with the material structure $\mathbf{C}_0$ in \eqref{eq:pa1} with $D$ and $\Omega$ given above, the elastic system \eqref{eq:lame1} becomes
\begin{equation}\label{eq:calr1}
  \left\{
    \begin{array}{ll}
      \mathcal{L}_{\breve{\lambda},\breve{\mu}}\bu(\bx) + \omega^2\bu(\bx)  =0 , & \mbox{in} \ \ D , \medskip  \\
      \mathcal{L}_{\hat{\lambda},\hat\mu}\bu(\bx)+ \omega^2\bu(\bx)  =0 , & \mbox{in} \ \ \Omega\backslash \overline{D} ,\medskip   \\
      \mathcal{L}_{\lambda, \mu}\bu(\bx) + \omega^2\bu(\bx)   =\bff, &  \mbox{in} \ \ \mathbb{R}^2\backslash \overline{\Omega},\medskip \\
      \bu|_- = \bu|_+, \quad \partial_{\breve{\bnu}}\bu|_- = \partial_{\hat{\bnu}}\bu|_+   & \mbox{on} \ \ \partial D ,\medskip\\
      \bu|_- = \bu|_+, \quad  \partial_{\hat{\bnu}}\bu|_- = \partial_{\bnu}\bu|_+ & \mbox{on} \; \partial \Omega.
    \end{array}
  \right.
\end{equation}

With the help of the potential theory, the solution to the equation system \eqref{eq:calr1} can be represented by 
\begin{equation}\label{eq:sc1}
  \bu(\bx)=
  \left\{
    \begin{array}{ll}
      \breve{\bS}^\omega_{\partial D}[\bvarphi_1](\bx), & \bx\in D,\medskip \\
      \hat{\bS}^\omega_{\partial D}[\bvarphi_2](\bx) + \hat{\bS}^\omega_{\partial\Omega}[\bvarphi_3](\bx), & \bx\in \Omega\backslash \overline{D},\medskip \\
      \bS^\omega_{\partial\Omega}[\bvarphi_4](\bx) + \mathbf{F}(\bx), & \bx\in \mathbb{R}^2\backslash \overline{\Omega},
    \end{array}
  \right.
\end{equation}
where $\bvarphi_1, \bvarphi_2, \bvarphi_3, \bvarphi_4\in L^2(\partial D)^2$ and $\mathbf{F}$ is the Newtonian potential of the source $\bff$ defined in \eqref{eq:newpf}. One can easily see that the solution given \eqref{eq:sc1} satisfies the first three condition in \eqref{eq:calr1} and the last two conditions on the boundary yield that 
\begin{equation}\label{eq:sc2}
  \left\{
    \begin{array}{ll}
      \breve{\bS}^\omega_{\partial D}[\bvarphi_1]=\hat{\bS}^\omega_{\partial D}[\bvarphi_2] + \hat{\bS}^\omega_{\partial \Omega}[\bvarphi_3], & \mbox{on} \quad \partial D,\medskip \\
       \partial_{\breve{\bnu}}\breve{\bS}^\omega_{\partial D}[\bvarphi_1|_- = \partial_{\hat{\bnu}}(\hat{\bS}^\omega_{\partial D}[\bvarphi_2] + \hat{\bS}^\omega_{\partial \Omega}[\bvarphi_3])|_+ , & \mbox{on} \quad \partial D, \medskip \\
      \hat{\bS}^\omega_{\partial D}[\bvarphi_2] + \hat{\bS}^\omega_{\partial \Omega}[\bvarphi_3]= \bS^\omega_{\partial \Omega}[\bvarphi_4] + \mathbf{F}, & \mbox{on} \quad \partial\Omega, \medskip \\
      \partial_{\hat{\bnu}}(\hat{\bS}^\omega_{\partial D}[\bvarphi_2] + \hat{\bS}^\omega_{\partial \Omega}[\bvarphi_3])|_- = \partial_{\bnu}(\bS^\omega_{\partial \Omega}[\bvarphi_4] + \mathbf{F})|_+ , & \mbox{on} \quad \partial \Omega.
    \end{array}
  \right.
\end{equation}
With the help of the jump formual in \eqref{eq:jump}, the equation system \eqref{eq:sc2} further yields the following integral system,
\begin{equation}\label{eq:cma}
  \left[
    \begin{array}{cccc}
        \breve{\bS}^\omega_{\partial D} & -\hat{\bS}^\omega_{\partial D,i} & -\hat{\bS}^\omega_{\partial \Omega, i} & 0 \\
      -\frac{1}{2} +(\breve{\bK}_{\partial D}^{\omega})^*   & -\frac{1}{2} - (\hat{\bK}_{\partial D}^{\omega})^*  &-\partial_{\hat{\bnu}_i} \hat{\bS}^\omega_{\partial\Omega} & 0 \\
      0 &  \hat{\bS}^\omega_{\partial D, e} & \hat{\bS}^\omega_{\partial \Omega, e} & - \bS^\omega_{\partial \Omega} \\
      0 & \partial_{\hat\bnu_e}\hat{\bS}^\omega_{\partial D} & -\frac{1}{2} + (\hat{\bK}_{\partial\Omega}^{\omega})^* & -\frac{1}{2} -({\bK}_{\partial\Omega}^{\omega})^* \\
    \end{array}
  \right]
\left[
  \begin{array}{c}
    \bvarphi_1 \\
    \bvarphi_2 \\
    \bvarphi_3 \\
    \bvarphi_4 \\
  \end{array}
\right]=
\left[
  \begin{array}{c}
    0 \\
    0 \\
    \mathbf{F} \\
    \partial_{\bnu}\mathbf{F} \\
  \end{array}
\right],
\end{equation}
where $\partial_{\hat{\bnu}_i}$ and $ \partial_{\hat\bnu_e}$ signify the conormal derivatives on the boundaries of $D$ and $\Omega$, respectively.

Following similar arguments as those in the previous section, there exists $\epsilon>0$ such that when $\bx\in B_{r_e + \epsilon}$ the Newtonian potential $\bF$ can be written as: 
\begin{equation}\label{eq:FF2}
  \bF= \sum_{n\geq N} \left(  \frac{ \kappa_{1,n} k_s r_e}{n J_n(k_s r_e)}   \mathbf{Q}_n^i   \right),
\end{equation}
where $\kappa_{1,n}$ are the coefficients, the functions $\mathbf{Q}_n^i $ are defined in \eqref{eq:qi}, and $N$ is large enough such the the spherical Bessel and Hankel functions, $J_n(t)$ and $H_n(t)$, fulfil the asymptotic expansions shown in \eqref{eq:asj}. 
We would like to remark that the Newtonian potential $\bF$ only contains the term $\mathbf{Q}_n^i$. Indeed, one can also include the term $\mathbf{P}_n^i $ and the analysis will be similar. To ease the exposition, we only consider the case that the Newtonian potential $\bF$ contains the term $ \mathbf{Q}_n^i $ only.
From the expressions for the functions $\mathbf{Q}_n^i $ in \eqref{eq:qi}, one has that on the boundary $\partial B_R$: 
\begin{equation}\label{eq:f}
 \bF= \sum_{n\geq N}  \bb_n^t \bff_n,
\end{equation}
where 
\[
\bb_n= \left(
   \begin{array}{c}
     e^{\rmi n \theta} \bnu \\
     e^{\rmi n \theta} \bt \\
   \end{array}
 \right), \quad 
 \bff_n =  \left(
   \begin{array}{c}
    f_{1,n}  \\
    f_{2,n} \\
   \end{array}
 \right)=  \left(
   \begin{array}{c}
     \kappa_{1,n} \eta_{1,n}  \\
     \kappa_{1,n} \eta_{2,n}\\
   \end{array}
 \right),
\]
with 
\[
 \eta_{1,n}=2, \quad \eta_{2,n}= \frac{2\rmi k_s r_e J_{n}^{\prime}(k_s r_e)}{n J_n(k_s r_e)}.
\]
Moreover, from the identities in \eqref{eq:tpi}, one has that on the boundary $\partial B_{r_e}$:
\begin{equation}\label{eq:tf}
\partial_{\bnu} \bF=\sum_{n\geq N}  \bb_n^t \tilde{\bff}_n,
\end{equation}
where
\[
 \tilde{\bff}_n =  \left(
   \begin{array}{c}
    \tilde{f}_{1,n}  \\
    \tilde{f}_{2,n} \\
   \end{array}
 \right)= \left(
   \begin{array}{c}
     \frac{  \kappa_{1,n} \gamma_{1,n}  k_s r_e}{n J_n(k_s r_e)}   \\
    \frac{   \kappa_{1,n} \gamma_{2,n}  k_s r_e}{n J_n(k_s r_e)} \\
   \end{array}
 \right),
\]
with $\gamma_{i,n}$, $1\leq i \leq 2$ given in \eqref{eq:tpi} with $R$ replaced by $r_e$.

Lemmas \ref{lem:eisin} and \ref{lem:pnu} show that based on the basis $\left(e^{\rmi n \theta}\bnu, e^{\rmi n \theta}\bt \right)$, the operators in the system \eqref{eq:cma} have the following expressions:
\[
\begin{split}
&\breve{\bS}^\omega_{\partial D}= \Rcal_{11n}, \quad  \hat{\bS}^\omega_{\partial D,i} =  \Rcal_{12n}, \quad  \hat{\bS}^\omega_{\partial \Omega, i}=  \Rcal_{13n}, \quad  (\breve{\bK}_{\partial D}^{\omega})^* =  \Rcal_{21n}, \\
&(\hat{\bK}_{\partial D}^{\omega})^*=  \Rcal_{22n}, \quad  \partial_{\hat{\bnu}_i} \hat{\bS}^\omega_{\partial\Omega} =  \Rcal_{23n}, \quad  \hat{\bS}^\omega_{\partial D, e}=  \Rcal_{32n}, \quad  \hat{\bS}^\omega_{\partial \Omega, e} =  \Rcal_{33n}, \\
&\bS^\omega_{\partial \Omega}=  \Rcal_{34n}, \quad   \partial_{\hat\bnu_e}\hat{\bS}^\omega_{\partial D}  =  \Rcal_{42n}, \quad  (\hat{\bK}_{\partial\Omega}^{\omega})^* =  \Rcal_{43n}, \quad  ({\bK}_{\partial\Omega}^{\omega})^*  =  \Rcal_{44n}, \\
\end{split}
\]
where 
\[
\Rcal_{11n} =
 \left(
   \begin{array}{cc}
      \breve{\alpha}_{1ni} &  \breve{\alpha}_{3ni} \\
      \breve{\alpha}_{2ni} &  \breve{\alpha}_{4ni} \\
   \end{array}
 \right), \quad
\Rcal_{12n} =
 \left(
   \begin{array}{cc}
      \hat{\alpha}_{1ni} &  \hat{\alpha}_{3ni} \\
      \hat{\alpha}_{2ni} &  \hat{\alpha}_{4ni} \\
   \end{array}
 \right), \quad
 \Rcal_{13n} =
 \left(
   \begin{array}{cc}
      \hat{\eta}_{1ni} &  \hat{\eta}_{3ni} \\
      \hat{\eta}_{2ni} &  \hat{\eta}_{4ni} \\
   \end{array}
 \right),  
\]
\[
\Rcal_{21n} =
 \left(
   \begin{array}{cc}
      \breve{a}_{1ni} &  \breve{b}_{1ni} \\
      \breve{a}_{2ni} &  \breve{b}_{2ni} \\
   \end{array}
 \right), \quad
\Rcal_{22n} =
 \left(
   \begin{array}{cc}
      \hat{a}_{1ni} &  \hat{b}_{1ni} \\
      \hat{a}_{2ni} &  \hat{b}_{2ni} \\
   \end{array}
 \right), \quad
 \Rcal_{23n} =
 \left(
   \begin{array}{cc}
      \hat{\zeta}_{1ni} &  \hat{\zeta}_{3ni} \\
      \hat{\zeta}_{2ni} &  \hat{\zeta}_{4ni} \\
   \end{array}
 \right),  
\]

\[
\Rcal_{32n} =
 \left(
   \begin{array}{cc}
      \hat{\eta}_{1ne} &  \hat{\eta}_{3ne} \\
      \hat{\eta}_{2ne} &  \hat{\eta}_{4ne} \\
   \end{array}
 \right), \quad
\Rcal_{33n} =
 \left(
   \begin{array}{cc}
      \hat{\alpha}_{1ne} &  \hat{\alpha}_{3ne} \\
      \hat{\alpha}_{2ne} &  \hat{\alpha}_{4ne} \\
   \end{array}
 \right), \quad
\Rcal_{34n} =
 \left(
   \begin{array}{cc}
      \alpha_{1ne} &  \alpha_{3ne} \\
      \alpha_{2ne} &  \alpha_{4ne} \\
   \end{array}
 \right),  
\]

\[
 \Rcal_{42n} =
 \left(
   \begin{array}{cc}
      \hat{\zeta}_{1ne} &  \hat{\zeta}_{3ne} \\
      \hat{\zeta}_{2ne} &  \hat{\zeta}_{4ne} \\
   \end{array}
 \right) , \quad
\Rcal_{43n} =
 \left(
   \begin{array}{cc}
      \hat{a}_{1ne} &  \hat{b}_{1ne} \\
      \hat{a}_{2ne} &  \hat{b}_{2ne} \\
   \end{array}
 \right), \quad
\Rcal_{44n} =
 \left(
   \begin{array}{cc}
      a_{1ne} &  b_{1ne} \\
      a_{2ne} &  b_{2ne} \\
   \end{array}
 \right).
\]
In the above epxressions, $\breve{\alpha}_{jni}$ with $j=1,2,3,4$ are given in Lemmas \ref{lem:eisin} with $R$ replaced by $r_i$ and with $(\mu, \lambda)$ replaced by $(\breve{\mu}, \breve{\lambda})$, and $\breve{a}_{jni}$ $\breve{b}_{jni}$ with $j=1,2$ are given in Lemma \ref{lem:npr} with $R$ replaced by $r_i$ and with $(\mu, \lambda)$ replaced by $(\breve{\mu}, \breve{\lambda})$. The same principle holds for parameters $\hat{\alpha}_{jni}, \hat{a}_{jni}, \hat{b}_{jni}, \hat{\alpha}_{jne}, \hat{a}_{jne}, \hat{b}_{jne}, {\alpha}_{jne}, {a}_{jne}, {b}_{jne}$ and the other parameters are given as follows:
\[
  \hat{\eta}_{1 ni} = -\frac{\rmi \pi}{2\omega^2 r_i} \left( n^2 J_n(\hat{k}_s r_i) H_n(\hat{k}_s r_e) + \hat{k}_p^2 r_i r_e J^{\prime}_n(\hat{k}_p r_i) H^{\prime}_n(\hat{k}_p r_e)  \right),
\]
\[
  \hat{\eta}_{2 ni} = \frac{n \pi}{2\omega^2 r_i } \left(\hat{k}_s r_i J^{\prime}_n(\hat{k}_s r_i) H_n(\hat{k}_s r_e) + \hat{k}_p r_e J_n(\hat{k}_p r_i) H^{\prime}_n(\hat{k}_p r_e)  \right),
\]
\[
  \hat{\eta}_{3 ni} = -\frac{n \pi }{2\omega^2 r_i} \left(\hat{k}_s r_e J_n(\hat{k}_s r_i) H^{\prime}_n(\hat{k}_s r_e) + \hat{k}_p r_i J^{\prime}_n(\hat{k}_p r_i) H_n(\hat{k}_p r_e)  \right),
\]
\[
  \hat{\eta}_{4 ni} = -\frac{\rmi \pi}{2\omega^2 r_i} \left( \hat{k}_s^2 r_e r_i J^{\prime}_n(\hat{k}_s r_i) H^{\prime}_n(\hat{k}_s r_e) + n^2 J_n(\hat{k}_p r_i) H_n(\hat{k}_p r_e)  \right);
\]

\[
  \hat{\eta}_{1 ne} = -\frac{\rmi \pi}{2\omega^2 r_e} \left( n^2 J_n(\hat{k}_s r_i) H_n(\hat{k}_s r_e) + \hat{k}_p^2 r_i r_e J^{\prime}_n(\hat{k}_p r_i) H^{\prime}_n(\hat{k}_p r_e)  \right),
\]
\[
  \hat{\eta}_{2 ne} = \frac{n \pi}{2\omega^2 r_e } \left(\hat{k}_s r_e J_n(\hat{k}_s r_i) H^{\prime}_n(\hat{k}_s r_e) + \hat{k}_p r_i J^{\prime}_n(\hat{k}_p r_i) H_n(\hat{k}_p r_e)  \right),
\]
\[
  \hat{\eta}_{3 ne} = -\frac{n \pi }{2\omega^2 r_e} \left(\hat{k}_s r_i J^{\prime}_n(\hat{k}_s r_i) H_n(\hat{k}_s r_e) + \hat{k}_p r_e J_n(\hat{k}_p r_i) H^{\prime}_n(\hat{k}_p r_e)  \right),
\]
\[
  \hat{\eta}_{4 ne} = -\frac{\rmi \pi }{2\omega^2 r_e} \left( \hat{k}_s^2 r_e r_i J^{\prime}_n(\hat{k}_s r_i) H^{\prime}_n(\hat{k}_s r_e) + n^2 J_n(\hat{k}_p r_i) H_n(\hat{k}_p r_e)  \right);
\]

\[
\begin{split}
\hat{\zeta}_{1ni} = &\frac{\mathrm{i} \pi}{2 \omega^{2} r_2^{2}}\left(2 \hat{\mu} n^{2} J_{n}(\hat{k}_{s} r_1)\left(H_{n}(\hat{k}_{s} r_2 )-\hat{k}_{s} r_2 H_{n}^{\prime} (\hat{k}_{s} r_2 )\right)+  \right. \\
&\left. J_{n}^{\prime}(\hat{k}_{p} r_1) \hat{k}_{p} r_1 \left(H_{n} (\hat{k}_{p} r_2 )\left(\omega^{2} r_2^{2}-2 \hat{\mu} n^{2}\right)+2 \hat{k}_{p} \hat{\mu} r_2 H_{n}^{\prime} (\hat{k}_{p} r_2 )\right)\right),
\end{split}
\]
\[
\begin{split}
\hat{\zeta}_{2ni} = &\frac{n \hat{\mu} \pi}{2 \omega^{2} r_2^{2}}\left(H_n(\hat{k}_s r_2)  J_n(\hat{k}_s r_1) \left(2 n^2- \hat{k}_s^2  {r_2}^2\right)-  2 J_n^{\prime}(\hat{k}_p r_1)  \hat{k}_p  {r_1}\times \right. \\
&\left.   \left( H_n(\hat{k}_p r_2)- H_n^{\prime}(\hat{k}_p r_2)  \hat{k}_p  {r_2}\right) - 2 H_n^{\prime}(\hat{k}_s r_2)  J_n(\hat{k}_s r_1)  \hat{k}_s  {r_2} \right),
\end{split}
\]
\[
\begin{split}
\hat{\zeta}_{3ni} = &\frac{-n \pi}{2 \omega^{2} r_2^{2}}\left(H_n(\hat{k}_p r_2)  J_n(\hat{k}_p r_1) \left(2 \hat{\mu}  n^2- \omega^2  {r_2}^2\right) -  2\hat{\mu}  J_n^{\prime}(\hat{k}_s r_1)  \hat{k}_s  {r_1}\times \right. \\
&\left.   \left( H_n(\hat{k}_s r_2)- H_n^{\prime}(\hat{k}_s r_2)  \hat{k}_s  {r_2}\right) - 2\hat{\mu}  H_n^{\prime}(\hat{k}_p r_2)  J_n(\hat{k}_p r_1)  \hat{k}_p  {r_2} \right),
\end{split}
\]
\[
\begin{split}
\hat{\zeta}_{4ni} = &\frac{\mathrm{i} \hat{\mu} \pi}{2 \omega^{2} r_2^{2}}\left(2 n^{2} J_{n}(\hat{k}_{p} r_1)\left(H_{n}(\hat{k}_{p} r_2 )-\hat{k}_{p} r_2 H_{n}^{\prime} (\hat{k}_{p} r_2 )\right)+\right. \\
&\left.J_{n}^{\prime}(\hat{k}_{s} r_1) \hat{k}_{s} r_1\left(H_{n} (\hat{k}_{s} r_2 )\left(k_s^{2} r_2^{2}-2 n^{2}\right)+2 \hat{k}_{s}  r_2 H_{n}^{\prime} (\hat{k}_{s} r_2 )\right)\right);
\end{split}
\]

\[
\begin{split}
\hat{\zeta}_{1ne} = &\frac{\mathrm{i} \pi}{2 \omega^{2} r_1^{2}}\left(2 \hat{\mu} n^{2} H_{n}(\hat{k}_{s} r_2)\left(J_{n}(\hat{k}_{s} r_1 )-\hat{k}_{s} r_1 J_{n}^{\prime} (\hat{k}_{s} r_1 )\right)+\right. \\
&\left. H_{n}^{\prime}(\hat{k}_{p} r_2) \hat{k}_{p} r_2\left(J_{n} (\hat{k}_{p} r_1 )\left(\omega^{2} r_1^{2}-2 \hat{\mu} n^{2}\right)+2 \hat{k}_{p} \hat{\mu} r_1 H_{n}^{\prime} (\hat{k}_{p} r_1 )\right)\right),
\end{split}
\]
\[
\begin{split}
\hat{\zeta}_{2ne} = &\frac{n \hat{\mu} \pi}{2 \omega^{2} r_2^{2}}\left(J_n(\hat{k}_s r_1)  H_n(\hat{k}_s r_2) \left(2 n^2- \hat{k}_s^2  {r_1}^2\right) - 2 H_n^{\prime}(\hat{k}_p r_2)  \hat{k}_p  {r_2} \times \right. \\
&\left.  \left( J_n(\hat{k}_p r_1)- J_n^{\prime}(\hat{k}_p r_1)  \hat{k}_p  {r_1}\right) - 2 J_n^{\prime}(\hat{k}_s r_1)  H_n(\hat{k}_s r_2)  \hat{k}_s  {r_1} \right),
\end{split}
\]
\[
\begin{split}
\hat{\zeta}_{3ne} = &\frac{-n \pi}{2 \omega^{2} r_1^{2}}\left(J_n(\hat{k}_p r_1)  H_n(\hat{k}_p r_2) \left(2 \hat{\mu}  n^2- \omega^2  {r_1}^2\right) - 2\hat{\mu}  H_n^{\prime}(\hat{k}_s r_2)  \hat{k}_s  {r_2}\times \right. \\
&\left.   \left( J_n(\hat{k}_s r_1)- J_n^{\prime}(\hat{k}_s r_1)  \hat{k}_s  {r_1}\right) - 2\hat{\mu}  J_n^{\prime}(\hat{k}_p r_1)  H_n(\hat{k}_p r_2)  \hat{k}_p  {r_1} \right),
\end{split}
\]
\[
\begin{split}
\hat{\zeta}_{4ne} = &\frac{\mathrm{i} \hat{\mu} \pi}{2 \omega^{2} r_1^{2}}\left(2 n^{2} J_{n}(\hat{k}_{p} r_2)\left(H_{n}(\hat{k}_{p} r_1 )-\hat{k}_{p} r_1 H_{n}^{\prime} (\hat{k}_{p} r_1 )\right)+\right. \\
&\left.J_{n}^{\prime}(\hat{k}_{s} r_2) \hat{k}_{s} r_2\left(H_{n} (\hat{k}_{s} r_1 )\left(k_s^{2} r_1^{2}-2 n^{2}\right)+2 \hat{k}_{s}  r_1 H_{n}^{\prime} (\hat{k}_{s} r_1 )\right)\right).
\end{split}
\]
The density functions $\bvarphi_i$ with $i=1,2,3,4$ can be written as, 
\begin{equation}\label{eq:den}
 \bvarphi_i=\sum_{n=-\infty}^{\infty} \bb_n^t \bvarphi_{i,n},
\end{equation}
where 
\[
\bb_n= \left(
   \begin{array}{c}
     e^{\rmi n \theta} \bnu \\
     e^{\rmi n \theta} \bt \\
   \end{array}
 \right), \quad 
  \bvarphi_{i,n} = \left(
   \begin{array}{c}
      \psi_{i,1,n} \\
      \psi_{i,2,n} \\
   \end{array}
 \right), 
\]
and the coefficients $\psi_{i,j,n}$, $1\leq i\leq 4, j=1,2$ are needed to be determined from the system \eqref{eq:cma}. Based on the discussion above, the system \eqref{eq:cma} is equivalent to solving the system 
\begin{equation}\label{eq:cma1}
\mathbf{M}
\left[
  \begin{array}{c}
    \bvarphi_{1,n} \\
    \bvarphi_{2,n} \\
    \bvarphi_{3,n} \\
    \bvarphi_{4,n} \\
  \end{array}
\right]=
\left[
  \begin{array}{c}
    0 \\
    0 \\
    {\bff}_n \\
    \tilde{\bff}_n \\
  \end{array}
\right]
\end{equation}
where $\bff$ and $\tilde{\bff}$ are given in \eqref{eq:f} and \eqref{eq:tf}, respectively, and the matrix $\mathbf{M}$ is given by 
\[
\mathbf{M}=
\left[
    \begin{array}{cccc}
        \Rcal_{11n} & \Rcal_{12n} & \Rcal_{13n} & 0 \\
      \Rcal_{21n} & \Rcal_{22n} & \Rcal_{23n} & 0 \\
      0 & \Rcal_{32n} & \Rcal_{33n} & \Rcal_{34n} \\
      0 & \Rcal_{42n} & \Rcal_{43n} & \Rcal_{44n} \\
    \end{array}
  \right].
\]

\begin{thm}\label{thm:CALR}
Consider the configuration $(\mathbf{C}_0, \bff)$ where $\mathbf{C}_0$ is given in \eqref{eq:pa1} and the Newtonian potential $\bF$ of the source term $\bff$ has the expression shown in \eqref{eq:FF2}. If the parameters in $\mathbf{C}_0$ are chosen as follows:
\begin{equation}\label{con:002}
\breve{\mu}=\mu, \quad \hat{\mu}=-\frac{\lambda + \mu}{\lambda + 3\mu} + \rmi\delta + p_{n_0}, \quad\mbox{and} \quad\delta=\rho^{n_0},
\end{equation}	
for some $n_0$, where $p_{n_0}=\mathcal{O}(1/n_0)$ is chosen such that 
\begin{equation}\label{eq:con1}
\det \mathbf{M}=\Ocal \left( \rho^{4n_0} \right),
\end{equation}
 then anomalous localized resonance occurs if the source $\mathbf{f}$ is supported inside the critical radius $r_*=\sqrt{r_e^3/r_i}$. Moreover, if the source is supported outside $B_{r_*}$, then no resonance occurs.
\end{thm}

\begin{proof}
 If the parameters in $\mathbf{C}_0$ are chosen as in \eqref{con:002}, 
 solving the system \eqref{eq:cma1} and applying the asymptotic expansion in \eqref{eq:asn} yield that the coefficients $\bvarphi_{i,j,n} $, $2\leq i \leq 4$, $j=1, 2$ have the following asymptotic expression,
\begin{equation}\label{eq:co2}
\begin{split}
 \abs{\varphi_{2,1,n} }& \leq \abs{\frac{ \kappa_{1,n} \rho^n n^2}{(\delta^2 + \rho^{2n})} }\left(1 + \mathcal{O}\left(\frac{1}{n}\right)\right), \quad 
\abs{ \varphi_{2,2,n}} \leq \abs{\frac{\rmi\kappa_{1,n} \rho^n n^2}{(\delta^2 + \rho^{2n})} }\left(1 + \mathcal{O}\left(\frac{1}{n}\right)\right), \\
\abs{ \varphi_{3,1,n} }& \leq \abs{\frac{\rmi\kappa_{1,n} \delta }{(\delta^2 + \rho^{2n})} }\left(1 + \mathcal{O}\left(\frac{1}{n}\right)\right), \quad 
 \abs{\varphi_{3,2,n}} \leq \abs{\frac{\kappa_{1,n} \delta n }{(\delta^2 + \rho^{2n})} }\left(1 + \mathcal{O}\left(\frac{1}{n}\right)\right), \\
 \abs{ \varphi_{4,1,n} }& \leq \abs{ \frac{\rmi\kappa_{1,n} \delta n}{(\delta^2 + \rho^{2n})} }\left(1 + \mathcal{O}\left(\frac{1}{n}\right)\right), \quad 
 \abs{\varphi_{4,2,n} }\leq \abs{\frac{\kappa_{1,n} \delta n }{(\delta^2 + \rho^{2n})} }\left(1 + \mathcal{O}\left(\frac{1}{n}\right)\right),
\end{split}
\end{equation}
where if $n=n_0$, the equalities hold, and if $n\neq n_0$, the inequalities hold. 


First, we show that the polarition resonance could occur when the source is located inside the critical radius $r^*$. 
From \eqref{eq:sc1}, the displacement field $\bu$ to the system \eqref{eq:calr1} in the shell $\Omega\backslash\overline{D}$ can be represented as 
\begin{equation}
\begin{split}
\bu&=  \sum_{n\geq N} \left( \hat{\bS}^\omega_{\partial D}[\varphi_{2,1,n} e^{\rmi n \theta} \bnu + \varphi_{2,2,n} e^{\rmi n \theta} \bt ](\bx) + \hat{\bS}^\omega_{\partial\Omega}[\varphi_{3,1,n} e^{\rmi n \theta} \bnu + \varphi_{3,2,n} e^{\rmi n \theta} \bt](\bx)  \right),
\end{split}
\end{equation}
where  the coefficients $\varphi_{i,j,n_0} $, $2\leq i \leq 4$, $j=1, 2$ satisfy the asymptotic expansions in \eqref{eq:co2}.
Thus with the help of Green's formula, the dissipation energy $E(\bu)$ defined in \eqref{def:E} can be written as 
\begin{equation}\label{eq:E2}
\begin{split}
 E(\bu) =&  \Im P_{\hat{\lambda},\hat{\mu}}(\bu, \bu) =  \Im\left(  \int_{\partial\Omega} \partial_{\hat{\bnu}} \bu \overline{\bu}ds  -   \int_{\partial D} \partial_{\hat{\bnu}} \bu \overline{\bu}ds \right)\\
 \geq & \frac{\kappa_{1,n_0}^2 \delta } {\delta^2+\rho^{2n_0}}  \geq \kappa_{1,n_0}^2  \left( \frac{r_e}{r_i} \right)^{n_0}.
 \end{split}
\end{equation}
If the source $\bff$ is supported inside the critical radius $r_*=\sqrt{r_e^3/r_i}$, by \eqref{eq:FF2} and the asymptotic properties of $J_n(t)$ and $H_n(t)$ in \eqref{eq:asn}, one can verify that there exists $\tau_1\in\mathbb{R}_+$ such that
\begin{equation}\label{eq:lll2}
 \limsup_{n\rightarrow\infty} \left(\frac{\kappa_{1,n}}{r_e^n} \right)^{1/n}=\sqrt{\frac{r_i}{r_e^3}+\tau_1}.
\end{equation}
Combining  \eqref{eq:E2} and \eqref{eq:lll2}, one can obtain that
\begin{equation}\label{eq:ipr}
E(\bu) \geq  \left(\frac{r_i}{r_e}+\tau_1 r_e^2\right)^{n_0} \left( \frac{r_e}{r_i} \right)^{n_0},
\end{equation}
which exactly shows that the polariton resonance occurs, namely the condition \eqref{con:res} is fulfilled. 

Then we prove the boundedness of the solution $\bu$ when $|x|>r_e^2/r_i$; that is, the bounded condition \eqref{con:bou} is satisfied. From \eqref{eq:sc1} and \eqref{eq:den}, the displacement field $\bu$ in $\mathbb{R}^2\backslash \overline{\Omega}$  can be represented as 
\begin{equation}
\begin{split}
\bu=  &  \sum_{|n|\geq N} \left(  {\bS}^\omega_{\partial\Omega}[\bvarphi_{4,1,n} e^{\rmi n \theta} \bnu + \bvarphi_{4,2,n} e^{\rmi n \theta} \bt](\bx)  \right) + \bF(\bx).
\end{split}
\end{equation}
Moreover, from \eqref{eq:co2} and Theorem \ref{thm:sing}, one can obtain that
\begin{equation}\label{bound1}
 \begin{split}
   |\bu|  \leq  \sum_{|n|\geq N} \abs{\kappa_{1,n}}  \frac{r_e^{2n}}{r_i^{n}}  \frac{1}{r^{n}} + |\bF|\leq C,
 \end{split}
\end{equation}
when $|x|>r_e^2/r_i$. Thus from \eqref{eq:ipr} and \eqref{bound1}, one can directly conclude that the CALR could occur when the source is located inside the radius $r_*=\sqrt{r_e^3/r_i}$.

Next we consider the case when the source is supported outside the critical radius $r_*$. 
From \eqref{eq:FF2} and the asymptotic properties of $J_n(t)$ and $H_n(t)$ in \eqref{eq:asn}, one can show that there exists $\tau_2>0$ such that
\[
 \limsup_{n\rightarrow\infty} \left(\frac{\kappa_{1,n}}{r_e^n} \right)^{1/n}\leq\frac{1}{r_*+\tau_2},
\]
and the dissipation energy $E(\bu)$ can be estimated as follows
\[
 \begin{split}
   E(\bu) & \leq \sum_{n\geq N} \kappa_{1,n}^2  \left( \frac{r_e}{r_i} \right)^{n} \leq \sum_{n\geq N}  \left( \frac{1}{(r_*+\tau_2)^2} \frac{r_e}{r_i} \right)^{n}  \leq \sum_{n\geq N}  \left( \frac{1}{\left(\sqrt{r_e^3/r_i}+\tau_2\right)^2} \frac{r_e^3}{r_i} \right)^{n} \leq C, 
 \end{split}
\]
which means that the polariton resonance does not occur. This completes the proof.
\end{proof}

\begin{rem}
We can verify the condition  \eqref{eq:con1} numerically. For this, we choose the following parameters:
 \[
 \begin{split}
& n_0=25, \ \ \omega=5, \ \  r_i=0.8, \ \  r_e=1, \ \  \breve{\mu}=\breve{\lambda}=\lambda=\mu=1, \ \  \delta=(r_i/r_e)^{n_0}=0.0038.
\end{split}
 \]
From the values of the parameters $\omega$ and $r_e$, one can readily verify that this is the case beyond quasi-static approximation. The  value of $\abs{\det \mathbf{M}}$ given in \eqref{eq:con1} in terms of the parameter $p_{n_0}$ is plotted in Fig.~ \ref{fig:CALR}, which apparently demonstrates that the condition \eqref{eq:con1} is satisfied. 
\begin{figure}[t]
  \centering
 {\includegraphics[width=5cm]{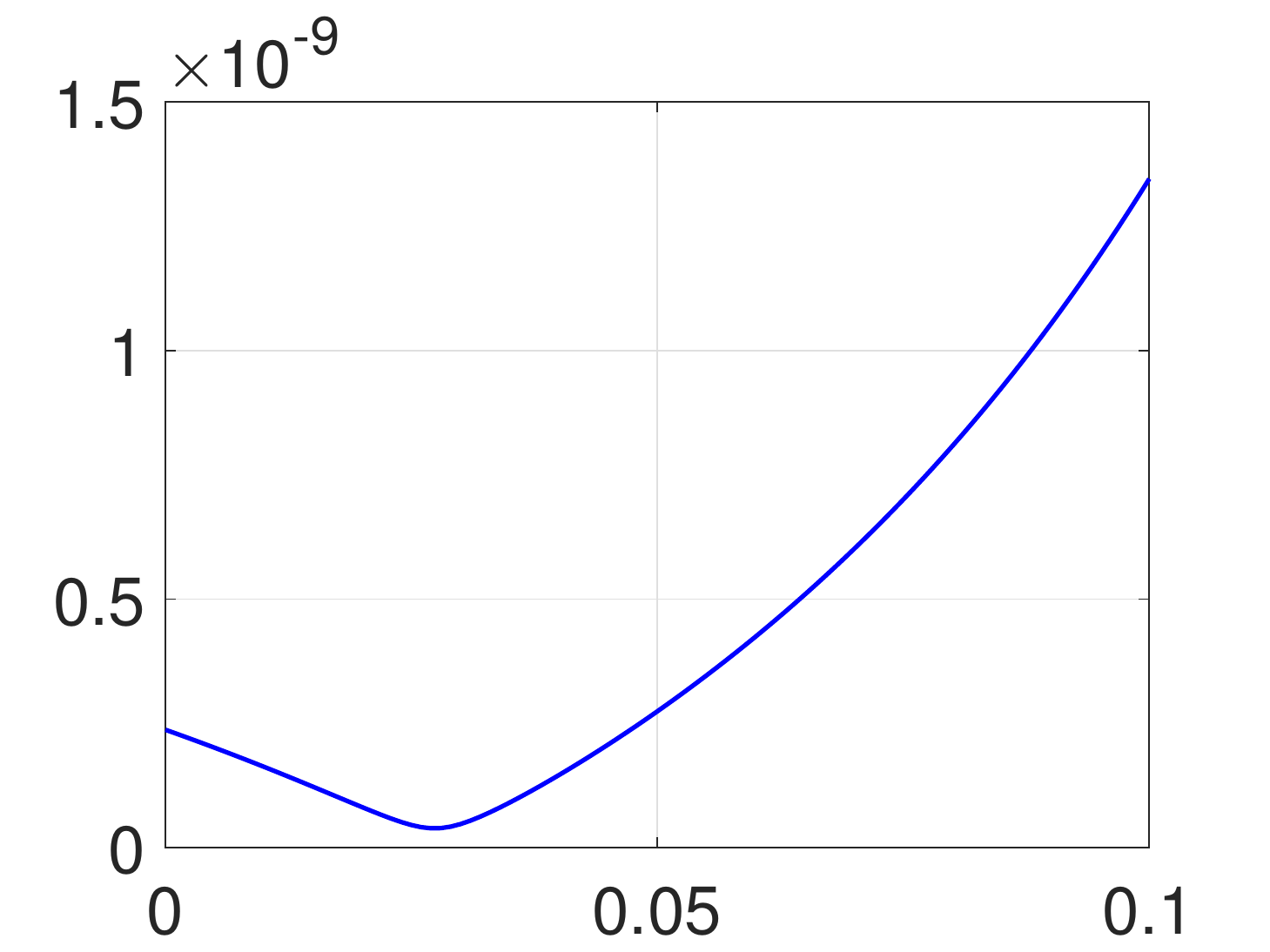}}
  \caption{The absolute value of $(\det \mathbf{M})$ given in \eqref{eq:con1} with respect to $p_{n_0}$. Horizontal axis:  value of $\abs{\det \mathbf{M}}$; Vertical axis: absolute value of $p_{n_0}$. }
  \label{fig:CALR}
\end{figure} 
\end{rem}

\section*{Acknowledgements}

The work of H. Li was supported by  Direct Grant for Research (CUHK).
The work of H. Liu was supported by the Hong Kong RGC General Research Fund (projects 12301420, 12302919, 11300821) and NSFC/RGC Joint Research Grant (project  N\_CityU101/21). 
The work of J. Zou was
substantially supported by Hong Kong RGC General Research Fund (projects 14306921 and
14306719).

\end{document}